\documentclass[10pt]{amsart}
\UseRawInputEncoding
\usepackage{amsfonts}
\usepackage{amssymb}
\usepackage{amsmath,amsthm}
\usepackage{amscd}
\usepackage{graphics}
\usepackage{amsmath,amssymb}
\usepackage{pdfsync}
\usepackage{mathrsfs}
\usepackage{stmaryrd}
\usepackage[colorlinks,linkcolor=blue,citecolor=blue,urlcolor=blue]{hyperref}
\newcommand{\rd}{\mathrm{d}}
\newcommand{\beq}{\begin{equation}}
\newcommand{\eeq}{\end{equation}}
\newcommand{\bea}{\begin{eqnarray}}
\newcommand{\eea}{\end{eqnarray}}
\newcommand{\nn}{\nonumber}
\newcommand\noi{\noindent}

\newcommand{\bk}{\begin{cases}}
\newcommand{\ek}{\end{cases}}

\newtheorem{theorem}{Theorem}

\newtheorem{corollary}[theorem]{Corollary}

\newtheorem{definition}[theorem]{Definition}
\newtheorem{example}[theorem]{Example}

\newtheorem{lemma}[theorem]{Lemma}

\newtheorem{proposition}[theorem]{Proposition}
\newtheorem{remark}[theorem]{Remark}

\begin{document}
\author{Piergiulio Tempesta}
\address{Departamento de F\'{\i}sica Te\'{o}rica II, Facultad de F\'{\i}sicas, Universidad
Complutense, 28040 -- Madrid, Spain
 and Instituto de Ciencias Matem\'aticas, C/ Nicol\'as Cabrera, No 13--15, 28049 Madrid, Spain.}
\email{p.tempesta@fis.ucm.es, piergiulio.tempesta@icmat.es}
\author{Giorgio Tondo}
\address{Dipartimento di Matematica e Geoscienze, Universit\`a  degli Studi di Trieste,
piaz.le Europa 1, I--34127 Trieste, Italy.}
\email{tondo@units.it}
\date{June 3, 2021}

\title[Haantjes Algebras of Classical Integrable Systems]{Haantjes Algebras of Classical \\ Integrable Systems}
\begin{abstract}

A tensorial approach to the theory of classical Hamiltonian integrable systems is proposed, based on the geometry of Haantjes tensors. We introduce the class of symplectic-Haantjes manifolds (or $\omega \mathscr{H}$ manifolds), as a natural setting where the notion of integrability can be formulated.
We prove that the existence of suitable Haantjes algebras of (1,1) tensor fields with vanishing Haantjes torsion is a necessary and sufficient condition for a Hamiltonian system to be integrable in the Liouville-Arnold sense.
We also show that new integrable models arise from the Haantjes geometry. Finally, we present an application of our approach to the study of the Post-Winternitz system and of a stationary flow of the KdV hierarchy.

\vspace{2mm}

\end{abstract}

 \maketitle
\tableofcontents

\section{Introduction}
Integrable systems are ubiquitous in many branches of modern mathematics and theoretical physics. Due to their relevance, in the last decades the search for intrinsic mathematical structures underlying the notion of integrability has been actively pursued. In particular, the investigation of the properties of exact solvability of integrable systems led to the discovery of new important analytic and geometric techniques. It is interesting to observe that finite-dimensional integrable models coming from classical or quantum mechanics share many geometric and algebraic properties with the infinite-dimensional ones described in terms of soliton equations.

The study of the geometry of classical integrable systems has a long history, dating back to the works by Liouville, Jacobi, St\"{a}ckel, Eisenhart, Arnold, etc.

In this context, the \textit{bi-Hamiltonian} approach has shown to be crucial for many respects. A bi-Hamiltonian  manifold is a differentiable manifold endowed with a pencil of Poisson structures \cite{Magri78}. In particular, the special class of \textit{$\omega N$ manifolds}, introduced in \cite{MM,FP}, is characterized by a non-degenerate Poisson bivector (whose inverse provides a symplectic structure $\omega$), and a compatible $(1,1)$ tensor field $\boldsymbol{N}$, also called recursion or hereditary operator. Such a tensor has a vanishing Nijenhuis torsion as a consequence of the existence of an underlying bi-Hamiltonian structure. The class of $\omega N$ manifolds offers a coherent approach to the construction of separation variables; it has been successfully applied, for instance, to the study of Gelfand-Zakharevich systems \cite{GZ,IMM,FP}.

The purpose of this paper is to present a new formulation of classical integrability based on Haantjes operators, namely  operator fields with vanishing \textit{Haantjes torsion}. The latter concept was introduced in 1955 by Haantjes in \cite{Haa} as a natural generalization of the Nijenhuis torsion \cite{GVY}. However,  the relevance of the Haantjes differential-geometric work in the realm of integrable systems quite surprisingly has not been recognized for a long time, with the exception of  some interesting applications to Hamiltonian systems of hydrodynamic type \cite{Fe,BogJMP,FeMa}.

The central notion underlying our formulation of integrability is that of \textit{Haantjes algebra}, introduced in \cite{TT2017}. Essentially, a Haantjes algebra is a pair $(M, \mathscr{H})$ where $M$ is a differentiable manifold and $\mathscr{H}$  is a set  of operator fields over $M$, with vanishing Haantjes torsion, which  satisfy suitable compatibility conditions among each others.

The study of these algebras was initiated by us for two reasons. Indeed,  they play a crucial role in the theory of diagonalization of operators on differentiable manifolds: Whenever the operators of a Haantjes algebra are semisimple and commute, a  set of local coordinates exists where all operators can be diagonalized simultaneously. Moreover, in the non-semisimple case, they acquire simultaneously a block-diagonal form. At the same time, Haantjes algebras naturally generalize several known interesting geometric structures arising in Riemannian geometry \cite{BCRframe,TT2015,TT2017}. A generalization of both the Haantjes torsion and the Haantjes theorem has been recently proposed in \cite{TT2021prepr}.

In this work, we will show the prominent role of Haantjes algebras in the theory of integrable systems. Indeed, we shall define a new family of manifolds, called  \textit{symplectic--Haantjes } (or $\omega\mathscr{H}$) manifolds. They are symplectic manifolds endowed with a Haantjes algebra of operators that are compatible with the symplectic structure.

 We shall prove that the integrability of a Hamiltonian finite-dimensional system can be characterized in terms of an Abelian algebra of  Haantjes operators, whose spectral and geometric properties turn out to be particularly rich. The notion of \emph{Haantjes chain}, defined in this framework, is a natural extension in the context of Haantjes geometry of previous similar notions known in the literature, as that of Lenard-Magri chain \cite{MagriLE} and generalized Lenard chain \cite{TKdV,FMT,MLenard}, relevant for quasi-bi-Hamiltonian systems and their generalizations.

A comparison between our notion of $\omega \mathscr{H}$ manifolds and   the recent definition of ``Haantjes manifolds'' due to Magri \cite{MFrob,MGall13,MGWDVV,MGall15} is in order. Our theory mainly differs from the fact that we assume the existence of an algebra of independent  Haantjes operators which are supposed to be compatible with the symplectic form $\omega$.  Besides, the Haantjes chains in our context are ``shorter'' than the ones defined in the recent Magri's theory \cite{MGall13}. This is due to a weaker
assumption that allows us to deal with both integrable and separable systems. This fact is an important novelty of the present work.

Our main result concerning integrability  is a theorem establishing that \textit{the existence of a $2n$-dimensional $\omega \mathscr{H}$ manifold is a necessary and sufficient condition for a non-degenerate
Hamiltonian system to be integrable} in the Liouville-Arnold sense. Precisely, we shall prove the existence of $n$ Haantjes operators

\begin{equation}\label{eq:LA}
\boldsymbol{K}_\alpha=\sum _{i=1}^n \frac{\nu_i^{(\alpha)}(\boldsymbol{J})}{\nu_i (\boldsymbol{J})}\bigg (\frac{\partial}{\partial J_i}\otimes \rd J_i +\frac{\partial}{\partial \phi_i}\otimes \rd \phi_i \bigg )\quad \alpha=1,\ldots,n
\ ,
 \end{equation}
where $\nu_i$ and $\nu_i^{(\alpha)}$ are the frequencies of the Hamiltonian $H$ and of the $(\alpha)-nth$ linear flow associated with the given system, respectively and $(\boldsymbol{J}, \boldsymbol{\phi})$ are a set  of action-angle variables. Formula \eqref{eq:LA} therefore intimately \textit{relates the Haantjes algebraic-geometric structure of an integrable system with its intrinsic dynamical properties}.

As a by-product of the main theorem, we will be able to define new general classes of integrable models possessing an
assigned Haantjes geometry. Quite interestingly, the systems so obtained are related to the wave equation.

An advantage of the present formulation \textit{\`a la Haantjes} (which also represents the main motivation for our study) is its generality: Haantjes tensors are indeed a larger class of tensors than those of Nijenhuis.

The proposed theory incorporates essentially all the known results on integrability of finite-dimensional systems that have been developed in a bi-Hamiltonian framework up to date, i.e. all the approaches based on Lenard chains and their generalizations (as quasi-bi-Hamiltonian systems \cite{Bl}, etc).

There is a neat relation between the Haantjes geometry developed in the present work and the well known Nijenhuis geometry. In fact, a  subfamily of symplectic-Haantjes manifolds
is provided by the class of symplectic-Nijenhuis ($\omega N$) manifolds. Precisely, we shall show that given an $\omega N$ manifold,  under mild assumptions one can  construct
an $\omega \mathscr{H}$ structure by taking $n$ independent powers of the recursion operator $\boldsymbol{N}$. In this case, $\boldsymbol{N}$ will play the role of a generator for the $\omega \mathscr{H}$ structure.

 However, it is important to notice that there exist $\omega \mathscr{H}$ structures not arising from a subjacent $\omega N$ structure. Indeed, whenever the Haantjes algebra $\mathscr{H}$ is non-Abelian, obviously it cannot be cyclically generated by a single Nijenhuis operator. This is the case, for instance, of the $\omega \mathscr{H}$ manifold associated with the superintegrable Post-Winternitz system.

Another noteworthy feature of our approach is that a priori the class of Haantjes algebras considered is \textit{not necessarily diagonalizable}. This aspect represents a generalization of the $\omega N$ approach, where indeed the operator $\boldsymbol{N}$ is  diagonalizable by hypothesis. Nevertheless, our theory keeps the intrinsic simplicity enjoyed by the standard approach to the Lenard-Magri chains for soliton hierarchies.

At the same time, the theory of $\omega \mathscr{H}$ manifolds is motivated by the crucial problem of the construction of coordinate systems allowing the additive separation of the associated Hamilton-Jacobi (HJ) equation (the separation variables). This is one of the most important problems in the theory of classical integrable systems, which historically has motivated a large amount of research work and inspired the formulation of several fundamental geometric developments.

The problem of the construction of separation of variables (SoV) can be recast in our approach and, in many cases, solved explicitly. Our main result in this direction is Theorem \ref{th:gHchart} ensuring the existence, under mild hypotheses, of a set of distinguished coordinates provided by the Haantjes structure associated with an integrable system, that we shall call the \textit{Darboux-Haantjes coordinates}. They represent separation coordinates for the Hamilton-Jacobi equation associated with the Hamiltonian functions of the given system.

The study of the general problem of separation of variables (including partial separation \cite{Stackel1897}, \cite{DKN2019}) in the Haantjes geometry is in progress; the case of multiseparable systems has been addressed in \cite{RTT2020}.

Finally, we mention that a generalization of $\omega \mathscr{H}$ manifolds that parallels the theory of Poisson-Nijenhuis manifolds \cite{KM} has been introduced in \cite{Thlt17}. The structures arising in this perspective, of Poisson-Haantjes type, will be suitable for studying Gelfand-Zakharevich systems \cite{GZ}.

The structure of the paper is as follows. In Section 2, we review the main algebraic properties of Nijenhuis and Haantjes tensors, and the notion of Haantjes algebras introduced in \cite{TT2017}.
In Section 3, we introduce the main geometric structures needed for the discussion of integrability, i.e. the $\omega \mathscr{H}$ manifolds; also, we clarify their relation with $\omega N$ manifolds.  In particular, a theorem guaranteeing the existence of the DH coordinates is proved. Besides, the notion of  generator of a Haantjes structure is defined. Section 4 contains the theorem that characterizes complete integrability via the Haantjes geometry. In Section 5, new integrable models related to the wave equation are deduced from suitable Haantjes structures. 

 In Section 6, a procedure for the construction of Haantjes structures for a given integrable system with two degrees of freedom is proposed. Also, the relevant example of the superintegrable Post-Winternitz system, whose separation coordinates are still not known,  is worked out.
 An application of our theory to a stationary reduction of the seventh-order equation of the KdV hierarchy is discussed in  Section 7. Some open problems are discussed in the final Section 8.

\section{Haantjes algebras of operators}
\label{sec:1}
Given a dynamical system defined over a finite--dimensional manifold $M$, a fundamental issue is to find suitable sets of coordinates  which allow us to decouple the equations of motion. The natural frames of such coordinates, being obviously integrable, can be characterized in a \textit{tensorial} manner as  eigendistributions of a suitable class of  $(1,1)$ tensor fields, i.e. the ones with a vanishing Nijenhuis or Haantjes tensor. In this section, we review some basic algebraic results concerning the theory of such tensors. For a more complete treatment, see
the original papers \cite{Haa,Nij}, the related ones \cite{Nij2,FN} and the recent review \cite{K17}.

\subsection{Preliminaries}
We shall denote by $M$  a differentiable manifold and by $\boldsymbol{L}:TM\rightarrow TM$ a $(1,1)$ smooth tensor field, i.e. a smooth field of linear operators on the tangent space at each point of $M$. In the following, all tensors will be considered to be smooth.
\begin{definition}
The
 \textit{Nijenhuis torsion} of $\boldsymbol{L}$ is the skew-symmetric  $(1,2)$ tensor field defined by
\begin{equation} \label{eq:Ntorsion}
\mathcal{T}_ {\boldsymbol{L}} (X,Y):=\boldsymbol{L}^2[X,Y] +[\boldsymbol{L}X,\boldsymbol{L}Y]-\boldsymbol{L}\Big([X,\boldsymbol{L}Y]+[\boldsymbol{L}X,Y]\Big),
\end{equation}
where $X,Y \in TM$ and $[ \ , \ ]$ denotes the commutator of two vector fields.
\end{definition}
In local coordinates $\boldsymbol{x}=(x^1,\ldots, x^n),$ the Nijenhuis torsion can be written in the form
\begin{equation}\label{eq:NtorsionLocal}
\begin{split}
(\mathcal{T}_{\boldsymbol{L}})^i_{jk}&=\sum_{\alpha=1}^n\bigg(\frac{\partial {\boldsymbol{L}}^i_k} {\partial x^\alpha} {\boldsymbol{L}}^\alpha_j -\frac{\partial {\boldsymbol{L}}^i_j} {\partial x^\alpha} {\boldsymbol{L}}^\alpha_k+\Big(\frac{\partial {\boldsymbol{L}}^\alpha_j} {\partial x^k} -\frac{\partial {\boldsymbol{L}}^\alpha_k} {\partial x^j}\Big ) {\boldsymbol{L}}^i_\alpha \bigg) \\
&= \sum_{\alpha=1}^n\bigg(
\boldsymbol{L}^\alpha_{[j}\partial_{|\alpha |}{\boldsymbol{L}}^i_{k]}-
\boldsymbol{L}^i_\alpha\partial_{[j} \boldsymbol{L}^\alpha_{k]}
\bigg) \ ,
\end{split}
\end{equation}
for the components of $\mathcal{T}_ {\boldsymbol{L}}$; among them, $n^2(n-1)/2$ are independent. Here for the sake of brevity we have used the notation $\partial_j:=\frac{\partial}{\partial x^j}$ and the indices between square brackets are to be skew--symmetrized, except those in $|\cdot|$.
\begin{definition}
\noi The \textit{Haantjes tensor} associated with $\boldsymbol{L}$ is the $(1, 2)$ tensor field defined by
\begin{equation} \label{eq:Haan}
\mathcal{H}_{\boldsymbol{L}}(X,Y):=\boldsymbol{L}^2\mathcal{T}_{\boldsymbol{L}}(X,Y)+\mathcal{T}_{\boldsymbol{L}}(\boldsymbol{L}X,\boldsymbol{L}Y)-\boldsymbol{L}\Big(\mathcal{T}_{\boldsymbol{L}}(X,\boldsymbol{L}Y)+\mathcal{T}_{\boldsymbol{L}}(\boldsymbol{L}X,Y)\Big).
\end{equation}
\end{definition}
The skew-symmetry of the Nijenhuis torsion implies that the Haantjes tensor is also skew-symmetric. Its explicit intrinsic expression is
\begin{equation} \label{eq:HaanEx}
\begin{split}
\mathcal{H}_{\boldsymbol{L}}(X,Y)=\boldsymbol{L}^4 [X,Y]+[\boldsymbol{L}^2X,\boldsymbol{L}^2Y] -2\boldsymbol{L}^3\Big([X,\boldsymbol{L}Y]+[\boldsymbol{L}X , Y]\Big)+\\
+\boldsymbol{L}^2\Big( [X, \boldsymbol{L}^2 Y]+4\, [\boldsymbol{L}X,\boldsymbol{L}Y]+[\boldsymbol{L}^2X,Y]\Big)
-2 \boldsymbol{L}\Big([\boldsymbol{L}X,\boldsymbol{L}^2Y]+[\boldsymbol{L}^2X,\boldsymbol{L}Y]\Big)
\ .
\end{split}
\end{equation}
Its local expression in recursive form is
\begin{equation}\label{eq:HaanLocal}
(\mathcal{H}_{\boldsymbol{L}})^i_{jk}=  \sum_{\alpha,\beta=1}^n\bigg(
\boldsymbol{L}^i_\alpha \boldsymbol{L}^\alpha_\beta(\mathcal{T}_{\boldsymbol{L}})^\beta_{jk}  +
(\mathcal{T}_{\boldsymbol{L}})^i_{\alpha \beta}\boldsymbol{L}^\alpha_j \boldsymbol{L}^\beta_k-
\boldsymbol{L}^i_\alpha\Big( (\mathcal{T}_{\boldsymbol{L}})^\alpha_{\beta k} \boldsymbol{L}^\beta_j+
 (\mathcal{T}_{\boldsymbol{L}})^\alpha_{j \beta } \boldsymbol{L}^\beta_k \Big)
 \bigg) \ .
\end{equation}
More explicitly, we have
\begin{equation} \label{eq:HaanExCoord}
\begin{split}
(\mathcal{H}_{\boldsymbol{L}})^i_{jk} &=   \sum_{\alpha=1}^n
\bigg(-  (\boldsymbol{L}^3)^i_\alpha \partial_{[j } \boldsymbol{L}^\alpha_{k]}+
  \sum_{\beta=1}^n\Big(
(\boldsymbol{L}^2)^i_\alpha\big(2\boldsymbol{L}^\beta_{[j}  \partial_{|\beta|} \boldsymbol{L}^\alpha_{k]}-
 \boldsymbol{L}^\beta_{[j} \partial_{k]}  \boldsymbol{L}^\alpha_{\beta})\\
 &+(\boldsymbol{L}^2)^\alpha_{[j}\boldsymbol{L}^{\beta}_{k]}\partial_{\alpha} \boldsymbol{L}^i_{\beta}
- \boldsymbol{L}^i_\alpha(\boldsymbol{L}^2)^\beta_{[j}\partial_{|\beta|}\boldsymbol{L}^\alpha_{k]}-
2 \sum_{\gamma=1}^n
 \boldsymbol{L}^i_\alpha \boldsymbol{L}^\beta_{[j}\boldsymbol{L}^\gamma_{k]}\partial_{\beta}
 \boldsymbol{L}^\alpha_{\gamma}
\Big)\bigg)
\ .
\end{split}
\end{equation}
The following notion is at the heart of the theory we are going to develop.
\begin{definition}
A Haantjes (Nijenhuis)   operator is an   operator  field whose  Haantjes (Nijenhuis) torsion identically vanishes.
\end{definition}
\begin{remark}\label{rem:Ldiag}
Any operator field on a two dimensional manifold  is  a Haantjes operator. Furthermore,  on an $n$-dimensional manifold, any operator field that admits local charts where it takes a diagonal form is also a Haantjes operator.
\end{remark}
\subsection{Haantjes algebras}

The concept of Haantjes algebra was defined in \cite{TT2017} as an abstract setting for developing the theory of Haantjes operators. Indeed, many important
properties of this class of operators can be discussed in a general, basis-independent context, in which the algebraic structure is kept to a minimum. As we shall see, additional structures
(as symplectic or Poisson structures) are needed if one wishes to discuss integrable models and the construction of suitable separation variables.
For sake of completeness, here a very brief review of the theory of Haantjes algebras is presented.

\begin{definition}\label{def:HM}
A Haantjes algebra of rank $m$ is a pair    $(M, \mathscr{H})$ which satisfies the following  conditions:
\begin{itemize}
\item
$M$ is a differentiable manifold of dimension $n$;
\item
$ \mathscr{H}$ is a set of Haantjes  operators $\boldsymbol{K}:TM\rightarrow TM$   that generates
\begin{itemize}
\item
 a  free  module  of rank $m$ over the ring of
smooth functions on $M$
\begin{equation}\label{eq:Hmod}
\mathcal{H}_{\big( f\boldsymbol{K}_1 +
                             g\boldsymbol{K}_2\big)}(X,Y)= \boldsymbol{0}
 \ , \qquad\forall X, Y \in TM \ ,\quad \forall f,g \in C^\infty(M)\ ,\quad \forall \boldsymbol{K}_1, \boldsymbol{K}_2 \in \mathscr{H},
\end{equation}
  \item
a ring  w.r.t. the composition operation
\begin{equation} \label{eq:Hring}
\mathcal{H}_{\big(\boldsymbol{K}_1 \, \boldsymbol{K}_2\big)}(X,Y)= \boldsymbol{0}\ , \qquad
\forall \boldsymbol{K}_1, \boldsymbol{K}_2 \in \mathscr{H}, \quad\forall X, Y \in TM\ ,
\end{equation}
\end{itemize}
\end{itemize}
If
\begin{equation} \label{eq:Kcomm}
\boldsymbol{K}_1\,\boldsymbol{K}_2=\boldsymbol{K}_2\,\boldsymbol{K}_1\, \qquad\forall \boldsymbol{K}_1, \boldsymbol{K}_2 \in \mathscr{H}
\end{equation}
the  algebra $\mathscr{H}$ will be said to be an Abelian Haantjes algebra. Moreover, if   the identity operator $\boldsymbol{I}\in \mathscr{H}$, then $(M, \mathscr{H})$ will be said to be a Haantjes algebra with identity.
\end{definition}
\par
The assumptions \eqref{eq:Hmod},  \eqref{eq:Hring} ensure that the set $\mathscr{H}$ is an \textit{associative algebra of Haantjes operators}; moreover, the Hamilton-Cayley theorem implies that its
rank $m$ is not greater than $n$.

\vspace{2mm}

The conditions of Definition \ref{def:HM} might seem difficult to realize. However, a natural class of  Haantjes algebras is given, in a local chart
$\{ U,\boldsymbol{x}=(x^1,\ldots,x^n)\}$,  by operators of the form
\begin{equation}\label{eq:Kdiag}
\boldsymbol{K}=\sum _{i=1}^n l_{i }(\boldsymbol{x})
\frac{\partial}{\partial x^i}\otimes \rd x^i \ ,
 \end{equation}
where $l_i(\boldsymbol{x}):= l^{i}_{i}(\boldsymbol{x})$ are arbitrary smooth functions playing the role of the eigenvalues of $\boldsymbol{K}$.
 The  diagonal operators \eqref{eq:Kdiag} have vanishing Haantjes tensor   and satisfy the differential compatibility condition \eqref{eq:Hmod}, by virtue of Remark \ref{rem:Ldiag}.  Furthermore, they form
a commutative ring; thus they also satisfy   Eqs. \eqref{eq:Hring}.
\begin{definition} \label{def:Kdiag}
The algebra generated by operators of the form \eqref{eq:Kdiag} will be said to be a \emph{diagonal} Haantjes algebra.
\end{definition}

\subsubsection{Cyclic Haantjes algebras}  \label{sec:cHa}
\par
 A particularly relevant class of Abelian Haantjes algebras  is given by
those generated by a {\it single} Haantjes operator   $\boldsymbol{L}:TM\mapsto TM$  \cite{TT2017}.
One can construct directly a Haantjes algebra $\mathcal{L}$ of  $rank$ $m\leq n=dim(M)$ by choosing as a set of generators  the   first $(n-1)$ powers of $\boldsymbol{L}$ together with $\boldsymbol{L}^0:=\boldsymbol{I}$
\begin{equation*}  \label{algebraic}
\mathcal{L}(\boldsymbol{L}):=Span\{ \boldsymbol{I} , \boldsymbol{L}, \ldots,\boldsymbol{L}^{n-1} ,\ldots\}=Span\{ \boldsymbol{I} , \boldsymbol{L},\ldots,\boldsymbol{L}^{n-1}\}\ .
\end{equation*}
The fact that $\mathcal{L}(\boldsymbol{L})$ is a Haantjes algebra is a consequence of the following result.
\begin{proposition} \label{pr:Lpowers}  \cite{BogI}.
Let $\boldsymbol{L}$ be an operator with vanishing  Haantjes tensor on $M$. Then, for any  polynomial $p(\boldsymbol{x},\boldsymbol{L})= \sum_{j=0}^{(n-1)} a_{j} (\boldsymbol{x}) \boldsymbol{L}^j$
with coefficients $a_{j}\in C^\infty(M)$, the associated Haantjes tensor also vanishes, i.e.
\begin{equation}
\mathcal{H}_{\boldsymbol{L}}(X,Y)= \boldsymbol{0} \ \Longrightarrow \
\mathcal{H}_{p(\boldsymbol{x},\boldsymbol{L})}(X,Y)= \boldsymbol{0}.
\end{equation}
\end{proposition}
\begin{proof}
See Corollary 3.3, p. 1136 of Ref. \cite{BogI}.
\end{proof}
According to the previous discussion, one can introduce the notion of cyclic Haantjes algebras.
\vspace{2mm}

\begin{definition} \label{def:CHa}
Let  $(M,  \mathscr{H})$ be an Abelian Haantjes algebra of rank $m$. An  operator $\boldsymbol{L}$, with minimal polynomial of degree $h\geq m$, is a generator of  $\mathscr{H}$  if
\begin{equation*}
 \mathscr{H}\subseteq\mathcal{L}(\boldsymbol{L})\ .
\end{equation*}
The corresponding algebra will be said to be a \emph{cyclic} Haantjes algebra.

Let
 \begin{equation} \label{eq:baseCicl}
 \mathcal{B}_{cyc}=\{ \boldsymbol{I} , \boldsymbol{L},\boldsymbol{L}^2,\ldots, \boldsymbol{L}^{m-1} \}
 \end{equation}
 be  a cyclic basis of $\mathcal{L}(\boldsymbol{L})$.  A basis $\mathcal{B}$ of  $ \mathscr{H}$ such that $\mathcal{B}\subseteq \mathcal{B}_{cyc}$ will be said to be a cyclic basis of $\mathscr{H}$.
\end{definition}
A generator of  $\mathscr{H} $ allows us to represent
each Haantjes operator $\boldsymbol{K}\in \mathscr{H} $  as a polynomial field in $\boldsymbol{L}$ of degree at most $(h-1)$, i.e.
\begin{equation}\label{eq:Hg}
\boldsymbol{K}=p_{\boldsymbol{K} }(\boldsymbol{x},\boldsymbol{L})=\sum_{i =0} ^{h-1} a_i(\boldsymbol{x})\,\boldsymbol{L}^i \ ,
\end{equation}
where $a_i(\boldsymbol{x})$ are  smooth functions on $M$.
A natural problem is to establish the conditions ensuring that a given Haantjes algebra is cyclic. This problem has been solved in Proposition 47 of \cite{TT2017}, where it was proven that each \textit{semisimple} Abelian algebra (see Definition \ref{def:Hss}) is cyclic. In the next subsection, we will present an extension of that result.

\subsection{Haantjes coordinates} \label{ss:Hcoor}
We shall always assume that the eigenvalues of any operator field  considered in this article are real functions. 
 
Let us recall the main result of \cite{TT2017} about the existence of Haantjes charts for Haantjes algebras.
The set of proper eigenvector fields of a generic operator $\boldsymbol{K}$, corresponding to an eigenvalue $l_i(\boldsymbol{x})$, is given by  $Ker(\boldsymbol{K}-l_{i } \boldsymbol{I})$,
whereas the set of its generalized eigenvector fields is $Ker(\boldsymbol{K}-l_{i } \boldsymbol{I})^{\rho_i}$. We have denoted by $\rho_i$ the Riesz index of $l_i$ (assumed to be independent of $\boldsymbol{x}$),
that is, the minimum integer such that
 $Ker(\boldsymbol{K}-l_{i } \boldsymbol{I})^{\rho_i}= Ker(\boldsymbol{K}-l_{i } \boldsymbol{I})^{\rho_i+1}$.
\begin{definition} \label{def:Hss}
A  Haantjes algebra is said to be semisimple if each $\boldsymbol{K} \in \mathscr{H}$ is semisimple (diagonalizable), that is, if each $\boldsymbol{K}$ admits a local reference frame formed by proper eigenvector fields of  $\boldsymbol{K}$.
\end{definition}

If a semisimple Haantjes algebra is Abelian, then there exists a local reference eigenframe in which all $\boldsymbol{K} \in  \mathscr{H}$ take simultaneously a diagonal form. A crucial question is to ascertain whether  an \emph{integrable} common eigenframe exists, or equivalently,
whether there exists a local coordinate chart where all operators can be simultaneously written in a  diagonal form.
The answer to this problem is offered by Theorem \ref{th:Halg}. Preliminarily, we  state the following
\begin{definition}
Let   $\{\mathcal{D}_i, \mathcal{D}_j, \ldots, \mathcal{D}_k \}$ be a set of distributions of vector fields. We shall say that these distributions  are mutually integrable if
\par
i) each of them is integrable;
\par
ii) any  sum $\mathcal{D}_i  + \mathcal{D}_j + \ldots + \mathcal{D}_k$ (where all indices $i,j,\ldots, k$ are different) is also integrable.
\end{definition}

\begin{theorem}\cite{TT2017} \label{th:Halg}
Let $(M, \mathscr{H})$ be an Abelian Haantjes algebra of rank $m$   and $\{\boldsymbol{K}_1, \ldots, \boldsymbol{K}_m \} $ a basis of it. Let us consider the spectral decomposition

\begin{equation} \label{eq:TMinters}
T_{\boldsymbol{x}}M= \bigoplus_{a=1}^v \mathcal{V}_a(\boldsymbol{x})
 \end{equation}
where
\begin{equation} \label{eq:Va}
\mathcal{V}_a(\boldsymbol{x}):=\mathcal{D}_{i_1}^{(1)}(\boldsymbol{x}) \bigcap  \ldots      \bigcap \mathcal{D}_{i_m}^{(m)}(\boldsymbol{x}) \qquad \qquad a:=(i_1,\ldots,i_m)
\end{equation}
and
\begin{equation}\label{def:Dialpha}
\mathcal{D}_{i_\alpha}^{(\alpha)}(\boldsymbol{x}) :=
Ker\big(\boldsymbol{K}_{\alpha}-l_{i_\alpha}^{(\alpha)}\boldsymbol{I}\big)^{\rho_{i_{\alpha}}}(\boldsymbol{x}), \qquad \alpha=1,\ldots,m , \ i_\alpha=1,\ldots,s_\alpha,
\end{equation}
where $s_\alpha$ is the number of the distinct eigenvalues of $\boldsymbol{K}_\alpha$.
The distributions $\mathcal{V}_a$ are mutually integrable; therefore,  there exists a set of  coordinates (that we shall call Haantjes coordinates) adapted to the decomposition \eqref{eq:TMinters},  such that all $\boldsymbol{K}\in \mathscr{H}$ can be simultaneously written in a  block-diagonal form.
Furthermore, if $\mathscr{H}$ is semisimple, in each set of Haantjes coordinates all $\boldsymbol{K}\in \mathscr{H}$ can be simultaneously written in a  diagonal form.
\end{theorem}
For the sake of clarity, we also quote Proposition 45 of  Ref. \cite{TT2017}, which is relevant in the forthcoming discussion.
\begin{proposition} \label{pr:CGK}
Let  $\boldsymbol{L}$ be a semisimple
operator with $h$ pointwise distinct eigenvalues $\{\lambda_{1}(\boldsymbol{x}), \ldots , \lambda_{h}(\boldsymbol{x})\}$, and $\boldsymbol{K}$ be another semisimple operator field possessing $s$ pointwise distinct eigenvalues, with $\text{s} \leq \text{h}$. The following conditions are equivalent:
\begin{itemize}
\item
 $\boldsymbol{K}$ belongs to the cyclic algebra of rank $h$ generated by $\boldsymbol{L}$, i.e.
\begin{equation} \label{eq:KinL}
 \boldsymbol{K}\in \mathcal{L}(\boldsymbol{L}) \ ;
\end{equation}
\item
 there exists a polynomial field  $p_{\boldsymbol{K} }(\boldsymbol{x},\lambda)$  in $\lambda$ of degree at most $(h-1)$ such that
\begin{equation} \label{eq:KpL}
 \boldsymbol{K}=p_{\boldsymbol{K} }(\boldsymbol{x},\boldsymbol{L}) \ ;
 \end{equation}
\item
each eigendistribution of  $\boldsymbol{L}$ is included in a single eigendistribution of  $\boldsymbol{K}$:
\begin{equation} \label{eq:cPayne}
\mathcal{C}_{\lambda_i}:= \ker(\boldsymbol{L}-\lambda_i \boldsymbol{I})
\subseteq\mathcal{D}_{l_i}:= \ker(\boldsymbol{K}-l_i \boldsymbol{I}),
\end{equation}
where it is understood that the eigenvalues $\Big(l_1(\boldsymbol{x}),\ldots,l_h(\boldsymbol{x})\Big)$ of $\boldsymbol{K}$ may not be all distinct.
\end{itemize}
\end{proposition}
We present now an extension of the result obtained in Proposition 47 of Ref. \cite{TT2017}, since it will be relevant in the formulation of the theory of $\omega\mathscr{H}$ manifolds.
\begin{proposition} \label{pr:CGKcoor}
Every    semisimple Abelian Haantjes algebra   $(M,\mathcal{H})$ of rank $m$ is cyclic and admits a  family of Haantjes generators; among them, there exists a Nijenhuis operator.

Moreover,  an operator $\boldsymbol{K} \in  \mathscr{H}$ is a generator of $\mathscr{H}$ if and only if  it possesses  $h$ distinct eigenvalues; precisely, $h=m$ if $\boldsymbol{I}\in \mathscr{H}$,  or $h=(m+1)$ otherwise.
 
\end{proposition}
\begin{proof}
Let us consider the Haantjes chart
\begin{equation}\label{eq:HchartKc}
\{ U, \boldsymbol{x}=(\boldsymbol{y}^{1}, \ldots,\boldsymbol{y}^v) \},
\end{equation}
adapted to the spectral decomposition \eqref{eq:TMinters}. Then,  
if  $v\geq m$, 
 each operator of the form
 \begin{equation}\label{eq:Ldiag}
\boldsymbol{L}=\sum_{a=1}^{v} \lambda_{a}(\boldsymbol{x})\sum_{j_{a}=1}^{r_{a}}
\frac{\partial}{\partial y^{a, j_{a}}}\otimes \rd y^{a, j_{a}}\ ,
 \end{equation}
where $r_a=rank~\mathcal{V}_a$ is a generator of $\mathscr{H}$, provided that its eigenvalue fields $\{\lambda_{1}(\boldsymbol{x}), \ldots , \lambda_{v}(\boldsymbol{x})\}$ are  arbitrary but
distinct smooth functions at any point of $U$. In fact, the eigendistributions of the operator \eqref{eq:Ldiag} are given by the distributions
$\mathcal{V}_a$ defined in Eq. \eqref{eq:Va};
consequently, by construction they satisfy condition \eqref{eq:cPayne}.  Besides, as the eigenvalues of $\boldsymbol{L}$  are distinct, this operator also satisfies the assumptions of Proposition \ref{pr:CGK}.
In particular, if the eigenvalues of $\boldsymbol{L}$ are chosen to be
\begin{equation} \label{eq:lambdaN}
\lambda_{a}(\boldsymbol{x})=\lambda_a(y^{a,1}, \ldots, y^{a,r_a}  )\qquad\qquad a=1,\ldots, v,
\end{equation}
then $\boldsymbol{L}$ is a Nijenhuis generator, that is, its Nijenhuis torsion identically vanishes.
If $v<m$, a generator can still be constructed, because we can further decompose some of the distributions $\mathcal{V}_{a}$ into a direct sum of mutually integrable sub-distributions. Precisely, we have
\[
\mathcal{V}_a=\left < \frac{\partial}{\partial y^{a, 1}},\ldots , \frac{\partial}{\partial y^{a, r_a}} \right >=\bigoplus_{i_a=1}^{\bar{r}_{i_{a}}} \left <  \frac{\partial}{\partial y^{a, 1}},\ldots,\frac{\partial}{\partial y^{a, i_a}}\right >=\bigoplus_{i_a=1}^{\bar{r}_{i_{a}}} \mathcal{C}_{a, i_a}\ ,
\]
with $\sum \bar{r}_{j_{a}}= r_a$; the previous decomposition of $\mathcal{V}_a$ can be realized in such a way that the number of addends appearing into the direct sum
\begin{equation} \label{eq:TMintersKC}
T_{\boldsymbol{x}}M=\bigoplus_{a=1,i_a=1}^{v,\bar{r}_{i_{a}}} \mathcal{C}_{a, i_a}
\end{equation}
is  not less than $ m$.

\vspace{2mm}

Let us recall that since   $ \mathscr{H}$  is a semisimple algebra by assumption, then the degree of the minimal polynomial of any $\boldsymbol{K}\in \mathscr{H}$ coincides with the number of its   distinct eigenvalues. 
\par
Assume now that  $\boldsymbol{K}$ belongs to $\mathscr{H}$ and has $h$ distinct eigenvalue fields,  with $h=m$ if $\boldsymbol{I}\in \mathscr{H}$,  or $h=(m+1)$ otherwise. 
If $\boldsymbol{I}\in \mathscr{H}$,  then $\mathcal{L}(\boldsymbol{K})\subseteq \mathscr{H}$ and having rank $m$, it coincides with $\mathscr{H}$.  Otherwise,   $\mathscr{H}=Span\{\boldsymbol{K}, \ldots, \boldsymbol{K}^m\} \subset \mathcal{L}(\boldsymbol{K})$.
 \par
 Conversely, assume  that a generator $\boldsymbol{L}$ belongs to $\mathscr{H}$. Then, by Lemma 37 of \cite{TT2017},  we have that $h\leq m$ if  $\boldsymbol{I}\in \mathscr{H}$, or $ h\leq m+1$  otherwise. Furthermore, we observe that if $\boldsymbol{I} \not\in  \mathscr{H}$ and $h=m$, then the cyclic Haantjes algebra $\mathcal{L}(\boldsymbol{L})$ would coincide with $\mathscr{H}$, which is absurd. Taking into account that  $h\geq m$ by Definition \ref{def:CHa}, the statement is proven.
\end{proof}
As in the case $v\geq m$ the number $v$ coincides with the number $h$ of the pointwise distinct eigenvalues of $\boldsymbol{L}$, we deduce the following
\begin{corollary} \label{cor:v=m}
Let $(M,\mathscr{H})$ be  a semisimple Abelian Haantjes algebra  of rank $\mathrm{m}$. Assume  that the number of the addends of the decomposition \eqref{eq:TMinters} is $v\geq m$. Then,   generators of $\mathscr{H}$ belonging to $\mathscr{H}$ exist if and only  if $v=m$, when $\boldsymbol{I}\in \mathscr{H}$, or    $v=m+1$, otherwise.
\end{corollary}

In Section \ref{sec:cOmH}, we shall specialize the results of Proposition \ref{pr:CGKcoor} and Corollary \ref{cor:v=m}  to the case of $\omega \mathscr{H}$ manifolds.

\subsection{Haantjes chains}
The theory of Lenard--Magri chains is a fundamental piece of the geometric approach to soliton hierarchies. Lenard--Magri chains have been introduced in order to construct integrals of motion in
involution for infinite-dimensional Hamiltonian systems \cite{Magri78,MagriLE} (see also  \cite{PS}, for a brief  history about the origin of the name ``Lenard chains'').
Besides, some non trivial generalizations of Lenard--Magri chains have proved to be useful in the study of separation of variables for finite-dimensional Hamiltonian systems (see \cite{MT,MTPLA,MTRomp,FMT,FP} and reference therein).

\noi Hereafter, we propose a further generalization of  the standard notions of the theory, which has the advantage to be both simple and
directly connected to the theory of  classical integrable systems.

\begin{definition}
 Let $( M,\mathscr{H})$ be a Haantjes algebra of rank $m$. We shall say that a smooth function $H$ generates a Haantjes chain of closed 1-forms of length $m$ if  there exist
a distinguished basis  $\{\tilde{\boldsymbol{K}}_1,\ldots, \tilde{\boldsymbol{K}}_m\}$ of $\mathscr{H}$
 such that
\begin{equation} \label{eq:MHchain}
\rd (\tilde{\boldsymbol{K}}^T_\alpha \,\rd H )=\boldsymbol{0} \ , \quad\qquad \alpha=1,\ldots ,m \ ,
\end{equation}
where $\tilde{\boldsymbol{K}}^{T}_{\alpha}: T^{*}M \to T^{*}M$ is the transposed operator of $\tilde{\boldsymbol{K}}_{\alpha}$ . The (locally) exact 1-forms 
$$
\rd H_\alpha=\tilde{\boldsymbol{K}}^T_\alpha \,\rd H 
$$
(which are supposed to be linearly independent), are called the elements of the Haantjes chain of length $\mathrm{m}$ generated by $H$; the functions $H_\alpha$ are their potential functions.
\end{definition}

In order to enquire about the existence of Haantjes chains for an assigned Haantjes algebra, we have to consider the codistribution, of rank  $r\leq m$, generated by a given function $H$ through an arbitrary basis
$ \{ \boldsymbol{K}_1,  \boldsymbol{K}_2,\ldots,\boldsymbol{K}_{m}\} $ of $\mathscr{H}$, i.e.
\begin{equation} \label{eq:codKH}
\mathcal{D}_H^\circ:=Span\{ \boldsymbol{K}_1^T dH,  \boldsymbol{K}_2^T \rd H,\ldots,\boldsymbol{K}_{m}^T\,\rd H\} \ ,
\end{equation}
and the distribution $\mathcal{D}_H$ of the vector fields annihilated by them, which has rank  $(n-r)$.
Note that such distributions do not depend on the particular basis chosen in $\mathscr{H}$.

The following theorem offers a geometric characterization of the existence of a Haantjes chain generated by a smooth function $H$ in terms of the Frobenius integrability of its
associated codistribution.
\begin{theorem}\label{th:LHint}
 Let  $(M,  \mathscr{H})$ be a  Haantjes algebra of rank $m$, and $H$ be a smooth function on $M$. Let  $\mathcal{D}_H^\circ$ be the codistribution \eqref{eq:codKH},
assumed to be of rank $m$ (independent on $\boldsymbol{x}$),
and $\mathcal{D}_H$ be the distribution of the vector fields annihilated by the $1$-forms of $\mathcal{D}_H^\circ$.
Then, the function $H$  generates a Haantjes chain $\eqref{eq:MHchain}$
if and only if $\mathcal{D}^\circ_H$ (or equivalently $\mathcal{D}_H$) is Frobenius-integrable.
\par
\end{theorem}
\begin{proof}
By definition, the Haantjes chain \eqref{eq:MHchain} contains $m$ independent  exact 1-forms. Therefore, they generate the integrable  distribution
\begin{equation} \label{eq:LHD}
 \mathcal{D}^\circ=Span \{\rd H_1=\tilde{\boldsymbol{K}}^T_1 \,\rd H ,\ldots, \rd H_m=\tilde{\boldsymbol{K}}^T_m\,\rd H\} \ ,
\end{equation}
which coincides with  $\mathcal{D}^\circ_H$.
\noi Vice versa, let $\mathcal{D}_H^\circ$ be integrable and  $\mathscr{D}_H$ its foliation. Then, there exist $m$ independent functions $(H_1,H_2,\ldots, H_m)$ which are constant on the leaves
of $\mathscr{D}_H$. Their differentials belong to $\mathcal{D}_H^\circ$, hence they can be written as
\begin{equation} \label{eq:Ltilde2}
\rd H_\alpha=
 \left(\sum_{\beta=1}^m a_{\alpha \beta} (\boldsymbol{x})\boldsymbol{K}_{\beta}^T\right )\rd H =: \tilde{\boldsymbol{K}}^T_\alpha \,\rd H \qquad \alpha=1,\ldots, m.
\end{equation}
\end{proof}

\section{The theory of symplectic-Haantjes manifolds}
In  this section we introduce the new class of symplectic-Haantjes manifolds;  we shall call them $\omega \mathscr{H}$ manifolds, by analogy with the well-known $\omega N$ manifolds \cite{MM,FP}.
These new manifolds are essentially Haantjes algebras endowed with a symplectic structure. The main reason to define these manifolds is that, apart from their interesting mathematical properties,   they provide a simple but sufficiently general setting in which the theory of Hamiltonian integrable systems can be naturally formulated.
\subsection{Haantjes algebras and $\omega \mathscr{H}$ manifolds}

\begin{definition}\label{def:oHman}
A symplectic--Haantjes (or $\omega \mathscr{H}$) manifold  of class $m$ is a triple $( M,\omega,\mathscr{H})$ which satisfies the following properties:
\begin{itemize}
\item[i)]
$(M,\omega)$  is a   symplectic  manifold of dimension $ 2 \, n$;
\item[ii)]
$\mathscr{H}$ is a Haantjes algebra of rank $m$;
\item[iii)]
$(\omega,\mathscr{H})$ are algebraically compatible, that is
$$
\omega(X,\boldsymbol{K} Y)=\omega(\boldsymbol{K} X,Y)  \qquad \forall \boldsymbol{K} \in \mathscr{H}\ ,
$$
or equivalently
\begin{equation}\label{eq:compOmH}
\boldsymbol{\Omega}\, \boldsymbol{K} =\boldsymbol{K}^T\boldsymbol{\Omega} ,\qquad \ \forall \boldsymbol{K} \in \mathscr{H}\ .
\end{equation}
\end{itemize}
\noi
Hereafter $\boldsymbol{\Omega}:=\omega ^\flat:TM\rightarrow T^*M$ denotes the  fiber bundles isomorphism defined by
$$
\omega(X,Y)=<\boldsymbol{\Omega} X,Y> \qquad\qquad\forall X, Y \in TM\ ,
$$
and the map $\boldsymbol{P}:=\boldsymbol{\Omega}^{-1}:T^*M \rightarrow TM$  is the Poisson bivector induced by  the symplectic structure $\omega$.
\par
If the identity operator $\boldsymbol{I}$ belongs to $\mathscr{H}$, then we shall say that $( M,\omega,\mathscr{H})$ is a $\omega \mathscr{H}$ manifold with identity. If $\mathscr{H}$ is an Abelian Haantjes algebra, the resulting $\omega \mathscr{H}$ manifold will be said to be Abelian.
\end{definition}

\par
\begin{remark}\label{rem:omHdiag}
The set of conditions of Definition \ref{def:oHman} admits a natural, simple realization. Indeed, note that in a coordinate system $\boldsymbol{x}=(x^1,\ldots,x^{2n})$
the operators
\begin{equation}\label{eq:Hdiaglambda}
\boldsymbol{K}_\alpha=\sum _{j=1}^{2n} \lambda_{j }^{(\alpha)}(\boldsymbol{x})
\frac{\partial}{\partial x^j}\otimes dx^{j}, \qquad\qquad \alpha=1,\ldots, m\ ,
 \end{equation}
according to Definition \ref{def:Kdiag}, generate a diagonal Haantjes algebra for any smooth function $\lambda_{j }^{(\alpha)}(\boldsymbol{x})$. Besides, by imposing the  algebraic compatibility conditions \eqref{eq:compOmH} we get  in a Darboux chart the solutions
$\{ \boldsymbol{x}=(q^1,\ldots q^n, p_1,\ldots,p_n)\}$ given by
\begin{equation}\label{eq:Hdiag}
\boldsymbol{K}_\alpha=\sum _{i=1}^{n} l_{i }^{(\alpha)}(\boldsymbol{x})
\Big(\frac{\partial}{\partial q^i}\otimes \rd q^i + \frac{\partial}{\partial p_i}\otimes \rd p_i\Big ), \qquad\qquad \alpha=1,\ldots, m\ .
 \end{equation}
Here $l_{i}^{(\alpha)}= \lambda_i^{(\alpha)}(\boldsymbol{x})=\lambda_{n+i}^{(\alpha)}(\boldsymbol{x})$, $i=1,\ldots,n$.
\end{remark}

As a consequence of conditions \eqref{eq:compOmH}, one can immediately deduce the following proposition, which turns out to be crucial for many results of the present theory. For instance, it has important consequences on the spectrum and the eigenvector fields of the Haantjes operators belonging to a $\omega\mathscr{H}$ manifold.

\begin{proposition}\label{pr:ss}
Let    $( M,\omega,\mathscr{H})$  be an $\omega \mathscr{H}$ manifold. Then,  any composed operator $\boldsymbol{\Omega}~p(\boldsymbol{x},\boldsymbol{K})$, $\boldsymbol{P}~q(\boldsymbol{x},\boldsymbol{K}^T)$ (where $p(\boldsymbol{x},\boldsymbol{K})$ and $q(\boldsymbol{x},\boldsymbol{K})$ are polynomial fields in $\boldsymbol{K}$ and $\boldsymbol{K}^{T}$ respectively) is skew-symmetric $\forall \boldsymbol{K} \in\mathscr{H}$. Moreover, if $\omega \mathscr{H}$ is Abelian, then $\boldsymbol{K}^T_{\alpha}\boldsymbol{\Omega} \boldsymbol{K}_{\beta}$,
 $\boldsymbol{K}_{\alpha} \boldsymbol{P} \boldsymbol{K}_{\beta}^T$,    are also skew-symmetric $\forall \boldsymbol{K}_{\alpha},\boldsymbol{K}_{\beta} \in\mathscr{H}$.
\end{proposition}

\begin{corollary}\label{lm:aupari}
Given a $2n$-dimensional $\omega \mathscr{H}$ manifold $M$, every generalized eigen--distribution  $Ker(\boldsymbol{K}-l_{i } \boldsymbol{I})^{{r_i}}$, $r_i \in \mathbb{N}$, is  of even rank.
Therefore, the geometric and algebraic multiplicities of  each   eigenvalue $l_{i }(\boldsymbol{x})$ are even.
\end{corollary}
\begin{proof}
 Given an $\omega \mathscr{H}$ manifold,  every generalized
 eigen--distribution $Ker(\boldsymbol{K}-l_{i } \boldsymbol{I})^{r_i}$  has the same rank of   the kernel of the operator
$\boldsymbol{\Omega} (\boldsymbol{K} -l_{i }  \boldsymbol{I})^{r_i}$, which is skew-symmetric by virtue of   Proposition \ref{pr:ss}. Thus, the second statement in the Corollary is a consequence of the fact that   the geometric multiplicity of
an   eigenvalue   $l_{i }(\boldsymbol{x})$ is  equal to the rank of $Ker(\boldsymbol{K}-l_{i } \boldsymbol{I})$,
 and its algebraic multiplicity  to the rank of $Ker(\boldsymbol{K}-l_{i } \boldsymbol{I})^{{\rho_i}}$.
  \end{proof}

Due to the above corollary, and the spectral decomposition of the tangent spaces of $M$ given by
\begin{equation}\label{def:Di}
T_{\boldsymbol{x}}M=\bigoplus_{i=1}^s \mathcal{D}_i(\boldsymbol{x}), \qquad \qquad
\mathcal{D}_i(\boldsymbol{x}) :=
Ker\big(\boldsymbol{K}-l_i\boldsymbol{I}\big)^{\rho_i}(\boldsymbol{x}),
\end{equation}
 as $rank$ $\mathcal{D}_i\geq 2, \  i=1,\ldots s$ we conclude that the number of the distinct eigenvalues of any
Haantjes operator $\boldsymbol{K}$ of an $\omega\mathscr{H}$ structure is not greater than $n$.
\begin{definition} \label{def:max}
Given a $2n$-dimensional $\omega\mathscr{H}$ manifold, if the number of  pointwise distinct eigenvalues of a Haantjes operator $\boldsymbol{K}\in \mathscr{H}$ is $n$, we shall say that such  operator is maximal.
\end{definition}
\noi
Observe that the minimal polynomial  of a maximal operator $\boldsymbol{K}\in \mathscr{H}$ has the form
\begin{equation} \label{def:minpol}
m_{\boldsymbol{K}}(\boldsymbol{x},\lambda)=
\prod_{i=1}^n \Big(\lambda - l_i(\boldsymbol{x})\Big)^{\rho_i} \ .
\end{equation}

\begin{lemma}\label{lm:max}
A Haantjes operator $\boldsymbol{K}$ of a $2n$-dimensional $\omega \mathscr{H}$ manifold of class $m$  
is maximal if and only if its minimal polynomial has degree  $m=n$. Therefore,  its minimal polynomial is the product of $n$ linear factors; thus,   $\boldsymbol{K}$ is pointwise semisimple.
\end{lemma}
\begin{proof}
As a consequence of  Corollary \ref{lm:aupari},  if $\boldsymbol{K}$ admits a Jordan chain of length $\rho_i$ associated with a given eigenvalue   $l_i(\boldsymbol{x})$,
there must exist, for the same eigenvalue, a ``twin'' Jordan chain of the same length. Consequently, the number of the Jordan chains of length $\rho_i$ associated to a given eigenvalue
is even, and  therefore $\rho_i\leq n$.   

Observe that $\boldsymbol{K}$ can be maximal  if and only if $\rho_i=1$, $i=1,\ldots,n$, due to Eq. \eqref{def:minpol}. In fact, in this case every eigendistribution of $\boldsymbol{K}$ has rank $2$ and it is formed by proper eigenvector fields.
\end{proof}

In the following,  we shall present some properties of both algebraic and differential-geometric nature for the Haantjes operators
  $\boldsymbol{K} \in \mathscr{H}$.
Denote by
\begin{equation} \label{eq:E}
\mathcal{E}_i := \bigoplus_{{j=1,\, j\neq i}} ^s  \mathcal{D}_j=Im \Big(\boldsymbol{K}-l_i\mathbf{I}\Big)^{\rho_i } \ ,\qquad\qquad i=1,\ldots,s,
\end{equation}
the  distribution of rank ($2 n-r_i$) spanned by all of the generalized eigenvectors of a Haantjes operator
$\boldsymbol{K}$, except those associated with the eigenvalue $l_i$.
Such a distribution  will be called a \emph{characteristic distribution} of  $\boldsymbol{K}$. Let $\mathcal{E}_i^{\circ}$ denote the module of one-forms that annihilate all vector fields of the distribution
$\mathcal{E}_i$.
\begin{proposition} \label{pr:relDE}
Given an $\omega \mathscr{H}$ manifold, the relations
\begin{eqnarray}
\label{eq:OmD}
\mathbf{\Omega}(\mathcal{D}_j)&=&\mathcal{E}_j^\circ \ \Leftrightarrow
\mathcal{D}_j = \boldsymbol{P}(\mathcal{E}_j^\circ)= \mathcal{E}_j^\perp  \ , \\
\label{eq:OmE}
\mathbf{\Omega}(\mathcal{E}_j)&=&\mathcal{D}_j^\circ\ \Leftrightarrow
\mathcal{E}_j = \boldsymbol{P}(\mathcal{D}_j^\circ)= \mathcal{D}_j^\perp \ ,
\end{eqnarray}
hold. Here $\mathcal{E}_j^\perp$ and $\mathcal{D}_j^\perp$ are the symplectic orthogonal distributions of $\mathcal{E}_j$ and $\mathcal{D}_j$, respectively.
\end{proposition}
\begin{proof}
Property \eqref{eq:OmD} follows from  the compatibility condition \eqref{eq:compOmH}, taking into account that the symplectic operator $\mathbf{\Omega}$ is invertible. In fact,
for each generalized eigenvector field $Y_j \in \mathcal{D}_j$, the one-form   $\mathbf{\Omega}Y_j$ is a generalized eigenform  of $\boldsymbol{K}^T$, as one infers from the relation
$$
(\boldsymbol{K}^T-l_j \boldsymbol{I})^{\rho_j}\boldsymbol{\Omega}\, Y_j\stackrel{Prop. \ref{pr:ss}}{=}\boldsymbol{\Omega} (\boldsymbol{K}-l_j \boldsymbol{I} )^{\rho_j}\,Y_j=\boldsymbol{0}.
$$
Then, since
\begin{equation} \label{eq:autof}
 Ker(\boldsymbol{K}^T-l_i \boldsymbol{I})^{\rho_i}=\bigg ( Im \Big(\boldsymbol{K}-l_i\mathbf{I}\Big)^{\rho_i } \bigg )^\circ=
\mathcal{E}_i^\circ \ ,
\end{equation}
 we deduce that $\mathbf{\Omega}Y_j(\boldsymbol{x})$ belongs to $\mathcal{E}_j^\circ(\boldsymbol{x})$. Since this subspace has the same dimension of  $\mathcal{D}_j(\boldsymbol{x})$, we obtain Eq. \eqref{eq:OmD}. The companion relation \eqref{eq:OmE} follows from Eq. \eqref{eq:OmD} jointly with the observation that,  by construction, $\mathcal{E}_j(\boldsymbol{x})$ is a complementary subspace of $\mathcal{D}_j(\boldsymbol{x})$ in $T_{\boldsymbol{x}}M$.
\end{proof}

\begin{proposition}\label{pr:FD}
Given an  $\omega \mathscr{H}$ manifold, the distributions $ \mathcal{D}_j$ of each $\boldsymbol{K} \in \mathscr{H}$ are integrable and of even rank.  Their  integral leaves are symplectic submanifolds of $M$ and are symplectically orthogonal to each other, namely
\begin{eqnarray}
\omega({\mathcal{D}_j,\mathcal{D}_j})&=&symplectic    \label{eq:Dsym} \\
\omega (\mathcal{D}_j,\mathcal{D}_k) & = & \boldsymbol{0}  \label{eq:Dsyo} \qquad\qquad j\neq k \ .
\end{eqnarray}
\end{proposition}
\begin{proof}
The distributions $ \mathcal{D}_j$ are integrable due to the Haantjes Theorem \cite{Haa} and are of even rank  by virtue of Corollary \ref{lm:aupari}. Besides,  they are symplectic as
$$
\mathcal{D}_j \cap( \mathcal{D}_j)^\perp \stackrel{\eqref{eq:OmE} }{=}
 \mathcal{D}_j\cap \mathcal{E}_j=\{\boldsymbol{0}\} \ .
$$
Finally, property \eqref{eq:Dsyo} follows from the fact that   $\mathcal{D}_k\subseteq \mathcal{E}_j\stackrel{\eqref{eq:OmE}}{\equiv} D_j^\perp$, if $j\neq k$.
\end{proof}
\begin{corollary}\label{pr:RomH}
Given an Abelian   $\omega \mathscr{H}$ manifold, each integral leaf $D_j$ of the eigen--distribution $ \mathcal{D}_j$,  $j=1,\ldots,s$, is an Abelian   $\omega \mathscr{H}$ submanifold.
\end{corollary}
\begin{proof}
As the Haantjes algebra $\mathscr{H}$ is Abelian by assumption, the submanifold $D_j$ is $\mathscr{H}$-invariant. Therefore, each operator $\boldsymbol{K}\in \mathscr{H}$ can be restricted to $D_j$, giving rise to  a Haantjes algebra of rank $\leq m$. Finally, the compatibility condition \eqref{eq:compOmH} holds since
$$
\boldsymbol{K}^T_{|D_j} \,\boldsymbol{\Omega}_{|D_j}=(\boldsymbol{K}^T\ \,\boldsymbol{\Omega})_{|D_j} =
({\boldsymbol{\Omega}}\,\boldsymbol{{K}})_{|D_j} =
\boldsymbol{\Omega}_{|D_j}\,\boldsymbol{{K}}_{|D_j} \qquad \forall \boldsymbol{{K}} \in  \mathscr{H} \ .
$$
\end{proof}
The $\omega \mathscr{H}$ manifold associated with the integral leaf $D_{j}$ will be denoted by
$(D_j, \omega_{|D_j}, \mathscr{H}_{|D_j})$.

Below, we shall prove that Proposition
 \eqref{pr:FD} and Corollary  \eqref{pr:RomH} hold for the intersections of the distributions $\mathcal{D}_{i_\alpha}^{(\alpha)} $
 defined in Eq. \eqref{def:Dialpha}. Precisely, let us consider the distributions
 \begin{equation} \label{eq:WaInt}
\mathcal{W}_{\mathcal{I}(\alpha)}:=\mathcal{W}_{i_1,\ldots,i_\alpha}:=\mathcal{D}_{i_1}^{(1)} \bigcap  \ldots      \bigcap \mathcal{D}_{i_\alpha}^{(\alpha)} \qquad\alpha=1,\ldots,m,
\end{equation}
which are integrable, being intersections of integrable distributions.
Interestingly, the integral leaves  of such distributions are themselves    $\omega \mathscr{H}$ manifolds.

\begin{proposition}\label{pr:FW}
Given  an Abelian   $\omega \mathscr{H}$ manifold, the distributions $ \mathcal{W}_{\mathcal{I}(\alpha)}$ \eqref{eq:WaInt} that have constant rank  are integrable and of even rank.  Their  integral leaves are symplectic submanifolds of $M$ and are symplectically orthogonal to each other, namely
\begin{eqnarray}
\label{eq:Vsym}
\omega \left(\mathcal{W}_{\mathcal{I}(\alpha)},\mathcal{W}_{\mathcal{I}(\alpha)}\right)&=&symplectic \\
\label{eq:VabsymOrt}
\omega \left(\mathcal{W}_{\mathcal{I}(\alpha)},\mathcal{W}_{\mathcal{I}(\beta)}\right) & = & \boldsymbol{0}   \qquad\qquad \alpha\neq \beta\ .
\end{eqnarray}
Moreover, they are  $\omega\mathscr{H}$ submanifolds of $(M,\omega,\mathscr{H})$.
\end{proposition}
\begin{proof}

We prove the result by induction on the number $\alpha$ of the factors in the intersection \eqref{eq:WaInt}. First, observe that if $\alpha=1$, properties  \eqref{eq:Vsym} and \eqref{eq:VabsymOrt}   are fulfilled according to Proposition \ref{pr:FD}. Let us suppose that these properties are satisfied by the distributions
   \begin{equation} \label{eq:Wam1}
\mathcal{W}_{\mathcal{I}(\alpha-1)}=\mathcal{W}_{i_1,\ldots,i_{\alpha-1}}=\mathcal{D}_{i_1}^{(1)}\bigcap  \ldots      \bigcap \mathcal{D}_{i_{(\alpha-1)}}^{(\alpha-1)}, \qquad \alpha >1 \ ;
\end{equation}
then we will prove that they hold true for the distributions  $\mathcal{W}_{\mathcal{I}(\alpha)}=\mathcal{W}_{\mathcal{I}(\alpha-1)} \bigcap  \mathcal{D}_{i_{\alpha}}^{(\alpha)}$. To this aim,  we consider the  foliation $\mathscr{W}_{\mathcal{I}(\alpha-1)}$ and a generic   leaf $W_{\mathcal{I}(\alpha-1)}$ of it.   Obviously,  the symplectic form $\omega$ can be restricted to $W_{\mathcal{I}(\alpha-1)}$, where it is still non degenerate, by the induction assumption \eqref{eq:Vsym}. Furthermore, as $W_{\mathcal{I}(\alpha-1)}$ is invariant for  the Haantjes algebra $ \mathscr{H}$, each element of $ \mathscr{H}$ can be restricted to $W_{\mathcal{I}(\alpha-1)}$.

Now, let us denote by   $ \mathscr{H}_{|W_{\mathcal{I}(\alpha-1)}}$ the restriction of the Haantjes algebra to $W_{\mathcal{I}(\alpha-1)}$. Over the leaf $W_{\mathcal{I}(\alpha-1)}$ the restricted compatibility condition holds true, as $\forall \boldsymbol{{K}} \in  \mathscr{H}$ we have
\begin{equation} \label{eq:OmHrest}
(\boldsymbol{K}^T)_{|W_{\mathcal{I}(\alpha-1)}}\ \,\boldsymbol{\Omega}_{|W_{\mathcal{I}(\alpha-1)}} =(\boldsymbol{K}^T\ \,\boldsymbol{\Omega})_{|W_{\mathcal{I}(\alpha-1)}} =
({\boldsymbol{\Omega}}\,\boldsymbol{{K}})_{|W_{\mathcal{I}(\alpha-1)}} =
\boldsymbol{\Omega}_{|W_{\mathcal{I}(\alpha-1)}}\,\boldsymbol{{K}}_{ |W_{\mathcal{I}(\alpha-1)}} \ .
\end{equation}
Therefore, the triple $(W_{\mathcal{I}(\alpha-1)},\omega_{|W_{\mathcal{I}(\alpha-1)}}, \mathscr{H}_{|W_{\mathcal{I}(\alpha-1)}})$ is a $\omega \mathscr{H}$ manifold.
\par
Notice that the generalized eigendistribution $\mathcal{W}_{\mathcal{I}(\alpha)}$ can be characterized at any point by $\mathcal{W}_{\mathcal{I}(\alpha)}(\boldsymbol{x})=
 Ker\big(\boldsymbol{{K}}_{\alpha}-l_{i_\alpha}^{(\alpha)}\boldsymbol{I}\big)^{\rho_i}_{|W_{\mathcal{I}(\alpha-1)}}(\boldsymbol{x})$.
 Consequently, the rank of $\mathcal{W}_\alpha$ must be even due to Corollary \ref{lm:aupari}. Moreover,
 \begin{eqnarray}
 \nn\mathcal{W}_{\mathcal{I}(\alpha)}^\perp(\boldsymbol{x}) &=& \left(Ker\big(\boldsymbol{{K}}_{\alpha}-l_{i_\alpha}^{(\alpha)}\boldsymbol{I}\big)^{\rho_i}_{|W_{\mathcal{I}(\alpha-1)}}\right)^\perp(\boldsymbol{x}) \\ \nn  &=&
\boldsymbol{ P}_{|W_{\mathcal{I}(\alpha-1)}}\left(\big(Ker\big(\boldsymbol{{K}}_{\alpha}-l_{i_\alpha}^{(\alpha)}\boldsymbol{I}\big)^{\rho_i}_{|W_{\mathcal{I}(\alpha-1)}}\big)^\circ\right)(\boldsymbol{x})
  = ( \mathcal{E}_{i_\alpha}^{(\alpha)})_{|W_{\mathcal{I}(\alpha)}} (\boldsymbol{x})\ ,
 \end{eqnarray}
 where $\boldsymbol{P}_{|W_{\mathcal{I}(\alpha-1)}}=( \boldsymbol{\Omega}_{|W_{\mathcal{I}(\alpha-1)}})^{-1}$.
 Therefore,
  \[
  \mathcal{W}_{\mathcal{I}(\alpha)} \bigcap  \mathcal{W}_{\mathcal{I}(\alpha)}^\perp= (\mathcal{D}_{i_{m}}^{(m)})_{|W_{\mathcal{I}(\alpha-1)}}  \bigcap ( \mathcal{E}_{i_\alpha} ^{(\alpha)})_{|W_{\mathcal{I}(\alpha-1)}}=\{\boldsymbol{0}\}
  \]
so that $\mathcal{W}_{\mathcal{I}(\alpha)}$ is a symplectic foliation of $M$. Observe that  condition \eqref{eq:OmHrest} also holds   over any leaf $W_{\mathcal{I}(\alpha)}$. Therefore, we conclude that
each integral leaf $W_{\mathcal{I}(\alpha)}$ of the distributions $ \mathcal{W}_{\mathcal{I}(\alpha)}$ is again a   $\omega \mathscr{H}$ manifold; it will be denoted by
$(W_{\mathcal{I}(\alpha)}, \omega_{|W_{\mathcal{I}(\alpha)}}, \mathscr{H}_{|W_{\mathcal{I}(\alpha)}})$.
 \par
To prove relation \eqref{eq:VabsymOrt}, let us compare $\mathcal{W}_{\mathcal{I}(\alpha)}=\mathcal{D}_{i_1}^{(1)} \bigcap  \ldots      \bigcap \mathcal{D}_{i_\alpha}^{(\alpha)} $ with $\mathcal{W}_{\mathcal{J}(\alpha)}=\mathcal{D}_{j_1}^{(1)} \bigcap  \ldots      \bigcap \mathcal{D}_{j_\alpha}^{(\alpha)} $. We distinguish two possibilities: If, $i_1\neq j_1, \ldots,i_\alpha\neq j_\alpha$, it follows that  $\mathcal{W}_{\mathcal{I}(\alpha)}\subseteq \mathcal{W}_{\mathcal{J}(\alpha)}^\perp$
as  $\mathcal{D}_{i_1}^{(1)}\subseteq   \mathcal{E}_{j_1}^{(1)}=\mathcal{D}_{j_1}^\perp, \ldots,  \mathcal{D}_{i_\alpha}^{(1)}\subseteq   \mathcal{E}_{j_\alpha}^{(\alpha)}= (\mathcal{D}_{j_\alpha}^{(\alpha)})^\perp$. Consequently, Eq. \eqref{eq:VabsymOrt} holds true. Instead, if some indices coincide, say
$i_1= j_1, \ldots, i_s=j_s$, but $i_{s+1}\neq j_{s+1},\ldots,i_\alpha\neq j_\alpha$, we have:
\[\mathcal{W}_{\mathcal{I}(\alpha)}= (\mathcal{D}_{i_{s+1}}^{(1)} \bigcap  \ldots      \bigcap \mathcal{D}_{i_\alpha}^{(\alpha)})_{| \mathcal{D}_{i_{1}}^{(1)} \bigcap  \ldots      \bigcap \mathcal{D}_{i_s}^{(s)}}, \quad \mathcal{W}_{\mathcal{J}(\alpha)}=( \mathcal{D}_{j_{s+1}}^{(1)} \bigcap  \ldots      \bigcap \mathcal{D}_{j_\alpha}^{(\alpha)})_{| \mathcal{D}_{i_{1}}^{(1)} \bigcap  \ldots      \bigcap \mathcal{D}_{i_s}^{(s)}}\ , \]
therefore we go back to the previous case. 

\end{proof}

\subsection{Darboux--Haantjes coordinates for $\omega\mathscr{H}$ manifolds}

We wish to show that among the families of  Haantjes coordinates defined in Section \ref{ss:Hcoor}, there exist sets of coordinates in which the symplectic form $\omega$ takes specifically the \textit{Darboux form}. In order to construct such Darboux-Haantjes coordinates, we first state the following

\begin{lemma}\label{lm:Omegaind}
Let $(M,\omega)$ be a symplectic manifold and let
\begin{equation} \label{eq:Hchart}
\{U,(x^{1,1},\ldots,x^{1,2 \sigma_{1}}; \ldots; x^{a,1},\ldots,  x^{a,2 \sigma_{a}};\ldots ; x^{v,1},\dots, x^{v,2 \sigma_{v}})\}
\end{equation}
be a local chart in $M$ in which $\omega$ takes locally the  block-form expression \begin{equation}\label{eq:HaanO}
\omega= \sum_{i,j=1}^{2 \sigma_1} \omega_{i j}^{(1)}  \rd x^{1,i} \wedge \rd   x^{1,j}+\ldots +
\sum_{i,j=1}^{2 \sigma_a}\omega_{i j}^{(a)}  \rd x^{a,i} \wedge \rd   x^{a,j}+ \ldots+
\sum_{i,j=1}^{2 \sigma_v}\omega_{i j}^{(v)}  \rd x^{v,i} \wedge \rd   x^{v,j}\ .
\end{equation}
Then, its components satisfy the equations
\begin{equation}\label{eq:omind}
\frac{\partial  \omega_{ij}^{(a)}}{\partial x^{b,k}}=0 \qquad\quad a\neq b\ , k=1,\ldots, 2 \sigma_b \ .
\end{equation}
\end{lemma}
\begin{proof}
It is an immediate consequence of the fact that $\omega$ is a closed form.
\end{proof}

\begin{theorem} \label{th:gHchart}
Let $(M,\omega,\mathscr{H})$ be an Abelian $\omega\mathscr{H}$ manifold of class $m$. Among the sets of Haantjes coordinates  for $\mathscr{H}$, there exist local  charts in $U\subset M$ with Darboux coordinates
\begin{equation} \label{eq:gDHch}
(\boldsymbol{q}^{1},\boldsymbol{p}_{1},\ldots,\boldsymbol{q}^{v},\boldsymbol{p}_{v}) \ ,
\end{equation}
where
$$
(\boldsymbol{q}^{1},\boldsymbol{p}_{1})=(q^{1,1},\ldots,q^{1,\sigma_{1}}, p_{1,1},\ldots,p_{1,\sigma_{1}}); \ldots;
 (\boldsymbol{q}^{v},\boldsymbol{p}_{v})=(q^{v,1},\dots, q^{v,\sigma_{v}},p_{v,1},\dots, p_{v,\sigma_{v}}),
$$
such that  $\omega$ takes the Darboux form
\begin{equation} \label{eq:gDo}
\omega=
\sum_{a=1}^{v}\sum_{i=1}^{\sigma_a}  \rd p_{a,i} \wedge \rd   q^{a,i} \ .
\end{equation}
Here $\sigma_{a}= 1/2$ $rank (\mathcal{V}_a)$, and $\mathcal{V}_a$ $(a=1,\ldots,v)$ are the distributions defined in Eq. \eqref{eq:Va}.
\end{theorem}
\begin{proof}
Let us consider a  Haantjes chart for $\mathscr{H}$. It is  adapted to the   decomposition \eqref{eq:TMinters} and has the form \eqref{eq:Hchart}. In such a chart, each element of the Haantjes algebra takes a
block-diagonal form, due to Theorem \ref{th:Halg}. Besides, the distributions $\mathcal{V}_a$ satisfy the properties stated in  Proposition \ref{pr:FW} as they correspond to the  distributions $\mathcal{W}_\alpha$ when $\alpha=m$. Consequently, in each Haantjes chart, the symplectic form $\omega$ takes the block form \eqref{eq:HaanO} as
$$
\omega\left(\frac{\partial}{\partial x^{a,i}},\frac{\partial}{\partial x^{b,j}}\right)=0\qquad\qquad
 a\neq b,\  i=1,\ldots,2\sigma_a, \ j=1,\ldots, 2 \sigma_b \ ,
$$
thanks to Eq. \eqref{eq:VabsymOrt}. Thus, its components satisfy property \eqref{eq:omind}. Then, over the  leaves of every distribution $\mathcal{V}_a$  one can find Darboux coordinates for the restriction of the symplectic form, which is still symplectic, due to Eq. \eqref{eq:Vsym}. Therefore, one can collect  such coordinates to obtain a local chart in $M$ like \eqref{eq:gDHch},
 adapted to the decomposition \eqref{eq:TMinters}. In this chart, the symplectic form $\omega$ takes  the Darboux form \eqref{eq:gDo} and each Haantjes operator $\boldsymbol{K} \in \mathscr{H}$ still possesses  a block-diagonal form.
  \end{proof}
\begin{definition}
Given an $\omega \mathscr{H}$ manifold, the local coordinates where all Haantjes operators take simultaneously a block-diagonal form and, at the same time, the symplectic form takes
the Darboux form \eqref{eq:gDo} are called Darboux--Haantjes (DH) coordinates.
\end{definition}

\begin{definition}
An $\omega \mathscr{H}$ manifold $(M, \omega, \mathscr{H})$ will be said to be semisimple if $ \mathscr{H}$ is a semisimple Haantjes algebra.
\end{definition}

\begin{corollary}
In a semisimple Abelian $\omega\mathscr{H}$ manifold $(M, \omega ,\mathscr{H})$,   each $\boldsymbol{K} \in \mathscr{H}$ takes the diagonal form \eqref{eq:Hdiag} in every set of  Darboux--Haantjes coordinates.
\end{corollary}
\begin{proof}
Given a semisimple Abelian Haantjes algebra $\mathscr{H}$, in each set of Haantjes coordinates the operators $\boldsymbol{K} \in \mathscr{H}$ take a diagonal form. Furthermore, due to Theorem \ref{th:gHchart}, among Haantjes coordinates there exist local charts in which the symplectic form takes the Darboux form \eqref{eq:gDo}. In such DH coordinates, every $\boldsymbol{K} \in \mathscr{H}$ takes the form \eqref{eq:Hdiag} as a consequence of the compatibility condition \eqref{eq:compOmH}.
\end{proof}

\subsection{Haantjes chains for $\omega\mathscr{H}$ manifolds}

The relevance of Haantjes chains in the theory of  $\omega\mathscr{H}$ manifolds is due to the following
\begin{lemma} \label{lm:MHchainInv}
 Let $(M,\omega,\mathscr{H})$ be an Abelian $\omega\mathscr{H}$ manifold. Then the potential functions $H_\alpha$,  whose differentials belong to all Haantjes chains  generated by a single function $H$, are in involution among each others and with $H$,
w.r.t. the Poisson bracket defined by the Poisson operator
 $\boldsymbol{P}=\boldsymbol{\Omega}^{-1}$.
\end{lemma}
\begin{proof}
In fact, we have
\begin{equation}
\{H_\alpha, H_\beta\}=<dH_\alpha, \boldsymbol{P}\, dH_\beta>=<\boldsymbol{K}_\alpha^T dH,\boldsymbol{P}\boldsymbol{K}^T_\beta dH>=
<dH,\boldsymbol{K}_\alpha \boldsymbol{P} \boldsymbol{K}_\beta^T dH>\stackrel{Prop.\, \ref{pr:ss}}{=}0,
\end{equation}
for  $\alpha,\beta=1,\ldots, m$.
The involution of $H_\alpha$ with $H$ can be proved  analogously.
\end{proof}

\par
An interesting problem is to study both Lagrangian eigendistributions and their associated Lagrangian foliations in an $\omega \mathscr{H}$ manifold. To this aim,
let us  compare  the distribution   $\mathcal{D}_H$, spanned by  the vector fields annihilated by the codistribution
$\mathcal{D}_H^\circ$ defined by Eq. \eqref{eq:codKH},
 with  the distribution, denoted by $\mathcal{D}^\perp_H$, of the vector fields symplectically orthogonal to those of $\mathcal{D}_H$. It is known that   $\mathcal{D}^\perp_H=\boldsymbol{P} (\mathcal{D}_H^\circ)
 $.
 Taking into
account Eq. \eqref{eq:codKH} and Proposition \ref{pr:ss},
it turns out that
\begin{equation}
 \mathcal{D}_H^\perp=Span\{ \boldsymbol{K}_1\, X_H,  \boldsymbol{K}_2\, X_H, \ldots , \boldsymbol{K}_{m}\, X_H\},
 \end{equation}
 where $X_H=\boldsymbol{P}\,\rd H$ is the Hamiltonian vector field with Hamiltonian function $H$.
Thus,   we deduce the following result.
\begin{proposition} \label{pr:LagD}
Let  $(M, \omega, \mathscr{H})$ be a $2n$-dimensional Abelian $\omega \mathscr{H}$ manifold, and $H$ be a smooth function on $M$. The relation
\begin{equation} \label{eq:DHortinclDH}
\mathcal{D}^\perp_H\subseteq \mathcal{D}_H \
\end{equation}
holds, namely, $\mathcal{D}_H$  is a coisotropic distribution  and  $\mathcal{D}^\perp_H$ is an isotropic one. Moreover, if
\[
rank(\mathcal{D}_H)= n,
\]
they  coincide and form a Lagrangian distribution.
\end{proposition}
\begin{proof}
Each vector field belonging to $\mathcal{D}_H^\perp$ is annihilated by any $1$-form belonging to $\mathcal{D}_H^\circ$ as
$$
<\boldsymbol{K}_{\alpha}^T \rd H ,  \boldsymbol{K}_{\beta} X_H>=<\rd H ,\boldsymbol{K}_{\alpha} \boldsymbol{K}_{\beta}\boldsymbol{P} \rd H>\stackrel{Prop. \ref{pr:ss}}{=}0, \qquad  \forall \boldsymbol{K}_{\alpha}, \boldsymbol{K}_{\beta}\in \mathcal{\mathscr{H}}.
$$
If $rank(\mathcal{D}_H)= n$, we also have $rank\big(\mathcal{D}_H^\circ\big )=n=
rank\big(P\big(\mathcal{D}_H^\circ)\big)$. Therefore,  $\mathcal{D}^\perp_H = \mathcal{D}_H$.
\end{proof}
\begin{proposition}\label{th:Lagr}
 Let  $(M, \omega,  \mathscr{H})$ be a $2n$-dimensional Abelian $ \omega \mathscr{H}$ manifold of class $m$ and $H$ be a smooth function that generates a Haantjes chain of length $m$. Then, the distribution
$\mathcal{D}_H$ (rs. $\mathcal{D}_H^\perp$) are integrable distributions and have a coisotropic (rs. isotropic)  foliation that we denote by $ \mathscr{D}_H$ (rs. $ \mathscr{D}_H^\perp$).  In particular, if $m = n$,
$ \mathscr{D}_H =
\mathscr{D}_H^\perp$ is a  Lagrangian foliation.
\par
\end{proposition}
\begin{proof}
Let  $\{H_1,H_2,\ldots, H_m\}$ be the potential functions of the Haantjes chain generated by $H$. They are in involution, due to Lemma \ref{lm:MHchainInv}; then, they are integral functions of  the coisotropic foliation  $\mathscr{D}_H$. Consequently, the isotropic distribution $\mathcal{D}_H^\perp$ is also integrable. In particular, if
  $m= n$, then  $rank(\mathcal{D}_H)=rank(\mathcal{D}_H^\perp)=n$. Thus, from Eq. \eqref{eq:DHortinclDH} it follows that $\mathcal{D}_H = \mathcal{D}_H^\perp$.
  \end{proof}
  The following theorem  clarifies the compatibility between a $\omega  \mathscr{H}$ manifold  and a set of functions in involution.
\begin{theorem} \label{pr:LHD}
Let   $(M, \omega,  \mathscr{H})$ be a
  $2n$-dimensional $\omega \mathscr{H}$ Abelian manifold of class $m$. Let $\{H_1,H_2,\ldots, H_m\}$
	be a set of independent functions  in involution and  $\mathcal{D}^\circ$ denote the codistribution spanned by their differentials.  The functions  $\{H_1,H_2,\ldots, H_m\}$  form  a Haantjes chain, generated by a smooth function $H$ in involution with  them,  if and only if $H$  satisfies the condition
  \begin{equation} \label{eq:LtH}
\mathcal{D}_H^\circ = \mathcal{D}^\circ\ .
  \end{equation}
\end{theorem}
\begin{proof}
Condition \eqref{eq:LtH} is equivalent to require that  $\mathcal{D}_H^\circ \subseteq \mathcal{D}^\circ$ as they both have, by assumption,
the same rank $m$. Thus, if such an inclusion is satisfied,  $\mathcal{D}_H$ is integrable and its foliation is  a coisotropic foliation due to Proposition \ref{pr:LagD}. Moreover,  by virtue of  Theorem \ref{th:LHint}, it follows that  the function $H$  generates   the Haantjes chain given by $\{\rd H_1, \rd H_2, \ldots, \rd H_m\}$.
\par
Conversely, if $\{H_1,H_2,\ldots, H_m \}$ are the potential functions of   a Haantjes chain generated by $H$, as a consequence of  Eq. \eqref{eq:MHchain} it follows that $\tilde{\boldsymbol{K}}_\alpha^T \rd H=\rd H_\alpha \in \mathcal{D}^\circ$, for $ \alpha=1, \ldots, m$. Then,  condition \eqref{eq:LtH} is satisfied.
\end{proof}

\subsection{Cyclic $\omega \mathscr{H}$  manifolds }  \label{sec:cOmH}
A particular, especially relevant family of Abelian $\omega \mathscr{H}$ manifolds  is represented by the class of
symplectic manifolds endowed with a cyclic Haantjes algebra of rank $m$ (see Sec. \ref{sec:cHa}). This algebra is generated by  a single Haantjes operator $\boldsymbol{L}$, assumed to satisfy the compatibility condition
\eqref{eq:compOmH}. Taking into account Eq. \eqref{eq:Hg}, it is easy to prove that such condition holds true also for each $\boldsymbol{K}\in \mathcal{L}(\boldsymbol{L})$. In fact, representing $\boldsymbol{K}$ as a polynomial field in $\boldsymbol{L}$, we have
$$
\boldsymbol{\Omega} \,\boldsymbol{K} =\boldsymbol{\Omega}\, p_{\boldsymbol{K} }(\boldsymbol{x},\boldsymbol{L}) =p_{\boldsymbol{K} }(\boldsymbol{x},\boldsymbol{L^T})\boldsymbol{\Omega}=\boldsymbol{K} ^T\boldsymbol{\Omega} \ .
$$
 We shall say that these manifolds are \emph{cyclic} $\omega  \mathscr{H}$ manifolds.
\par
For a  cyclic $\omega  \mathscr{H}$ manifold one can construct a special class of  Haantjes chains. Indeed, in this context, Theorem \ref{th:LHint} amounts to say that a function $H$ generates the
Haantjes chain
\begin{equation} \label{defi:gLc}
\rd H_\alpha =\boldsymbol{K}^T_\alpha dH= p_\alpha (\boldsymbol{L}^T) dH \qquad\qquad \alpha=1,\ldots,m
\end{equation}
if and only if the codistribution
\begin{equation}
\mathcal{D}_H^\circ =Span\{\rd H, \boldsymbol{L}^T \rd H, \ldots, (\boldsymbol{L}^T)^{m-1} \rd H\}
\end{equation}
is integrable.
The chain \eqref{defi:gLc} will be said to be a \emph{cyclic Haantjes  chain}.
\par
An important class  of cyclic $\omega \mathscr{H}$ manifolds is represented by    $\omega N$ manifolds  \cite{MM,Mnoi}. Let $(M,\omega,\boldsymbol{N})$ be a symplectic-Nijenhuis (or $\omega N$) manifold, that is,  a manifold endowed with a symplectic form $\omega$ and a Nijenhuis operator
$\boldsymbol{N}$ that satisfy the following compatibility conditions
\begin{eqnarray}
\label{eq:OmNa}
\boldsymbol{\Omega}\boldsymbol{N}-\boldsymbol{N}^T\boldsymbol{\Omega}&=& \boldsymbol{0} \ , \\
\label{eq:OmNd}
\rd(\boldsymbol{\Omega} \boldsymbol{N})&=&\boldsymbol{0} \ .
\end{eqnarray}
Here
\beq
\rd\boldsymbol{\Omega}(X,Y)=\mathcal{L}_X(\boldsymbol{\Omega})Y-\mathcal{L}_Y(\boldsymbol{\Omega} )X+\rd<\boldsymbol{\Omega} X,Y> +  \boldsymbol{\Omega}\, [X,Y]\ , \quad \forall X,Y\in TM
\eeq
and $\mathcal{L}_X$ denotes the Lie derivative of a tensor field with respect to the vector field $X$.
\begin{example} \label{ex:OmN}
Let us suppose  that the Nijenhuis operator $\boldsymbol{N}$ has its minimal polynomial  of degree $m$.
Then, the $\omega N$  manifold $M$ has a \emph{standard} $ \omega \mathscr{H}$ structure, given by
$$
(M, \omega,  \boldsymbol{K}_1= \boldsymbol{I},  \boldsymbol{K}_2= \boldsymbol{N},\ldots,   , \boldsymbol{K}_{m}= \boldsymbol{N}^{m-1}) \ ,
$$
with a Haantjes algebra of rank  $m\leq dim(M)$.
In fact, each Nijenhuis operator $ \boldsymbol{N}$ is also a Haantjes operator; therefore, it generates the  cyclic Haantjes algebra $\mathcal{L}(\boldsymbol{N})$. Besides, the algebraic compatibility condition \eqref{eq:OmNa} assures that for all Haantjes operators
\begin{equation} \label{eq:KOmN}
\boldsymbol{K}=p_{\boldsymbol{K} }(\boldsymbol{x},\boldsymbol{N})=\sum_{i =0} ^{m-1}  a_i(\boldsymbol{x})\,\boldsymbol{N}^i ,
\end{equation}
condition iii) of Definition \ref{def:oHman} is fulfilled.
 \par
 Furthermore, the differential  condition \eqref{eq:OmNd} implies that for all $\boldsymbol{K}$, the following relation 
 \begin{equation}
\rd (\boldsymbol{\Omega},\boldsymbol{K})(X,Y)= \sum_{i=0}^{m-1}\rd a_i \wedge (\boldsymbol{\Omega} \boldsymbol{N}^i)(X,Y) \qquad\qquad  \forall X,Y \in TM
\end{equation}
holds.
\end{example}
\par
Note that the cyclic  Haantjes chains  on $\omega N$ manifolds coincide with the notion of Nijenhuis chains  \cite{FMT} and  of generalized Lenard chains, defined in  \cite{TT,TGalli12}. In particular, if   $\tilde{ \boldsymbol{K}}_\alpha=\boldsymbol{N}^{\alpha-1} $, the cyclic  Haantjes chains coincide with  the classical Lenard-Magri chains \cite{Mnoi} (see also the analysis of the Benenti systems in terms of Killing-St\"ackel in  algebras \cite{BCRframe}).
\par

We shall now deepen into the relationship between the notion of cyclic Haantjes algebras and the theory of $\omega\mathscr{H}$ manifolds.

\begin{proposition} \label{pr:Hmg}
Every semisimple $2n$-dimensional Abelian $\omega \mathscr{H}$ manifold $(M,\omega,\mathscr{H})$ of class $m$ is a cyclic $\omega \mathscr{H}$ manifold. Moreover,  each generator  of $\mathscr{H}$ that belongs to $\mathscr{H}$ must have $h$ distinct eigenvalues, with $h=m$ if $\boldsymbol{I}\in \mathscr{H}$,  or
$h=m+1$ otherwise. In particular, if $\mathscr{H}$  has  rank $m=n$,  an operator $\boldsymbol{K}\in \mathscr{H}$  is a  generator of  $\mathscr{H}$ if and only if it is  maximal. In this case,  $\boldsymbol{I}\in \mathscr{H}$.
\par
Finally, if a  $2n$-dimensional $\omega \mathscr{H}$ manifold   of class $n$ is non-semisimple,  then none of its generators can be maximal.
\end{proposition}
\begin{proof}
 The first and the second statement are a direct consequence   of Proposition \ref{pr:CGKcoor}. In particular, if $h=n$, the cyclic algebra $\mathcal{L}(\boldsymbol{K})$ generated by a maximal operator $\boldsymbol{K} \in \mathscr{H}$  has rank $n$; therefore, it coincides with $\mathscr{H}$; thus,
$\boldsymbol{I}\in \mathscr{H}$.
Conversely, every generator of  $\mathscr{H}$, being semisimple and having its minimum polynomial of degree $n$ (by virtue of the second statement), is maximal.
\par
Finally, we observe that generators of non-semisimple Haantjes algebras cannot be maximal, since maximal operators,  due to Lemma \ref{lm:max}, are semisimple.
\end{proof}
The following Proposition, which specializes the results of Proposition \ref{pr:CGKcoor}  and Corollary \ref{cor:v=m}  to the case of a semisimple $2n$-dimensional Abelian $\omega \mathscr{H}$ manifold  of class $m$, presents an explicit construction of  a generator of $\mathscr{H}$; in particular, this generator  can be  chosen to be a Nijenhuis operator.

\begin{proposition} \label{th:HgDH}
Let  $(M,\omega,\mathscr{H})$ be an Abelian $2n$-dimensional semisimple $\omega\mathscr{H}$ manifold of class $m$. Let us consider the spectral decomposition \eqref{eq:TMinters}  and a Darboux-Haantjes chart $\{ U,(q^{a,j_a},p_{a,j_a}) \}$, $a=1, \ldots,v$, $j_a=1,\ldots , \sigma_a=\frac{1}{2} rank( \mathcal{V}_{a})$,  adapted to the decomposition \eqref{eq:TMinters}, namely
\begin{equation}\label{eq:HchartK}
\mathcal{V}_{a}=Span\left\{\frac{\partial}{\partial q^{a,j_a }},\frac{\partial}{\partial p_{a,j_a }}\right\} \ .
\end{equation}
Then, if $m\leq v\leq n$,  
each operator defined by
 \begin{equation}\label{eq:Lgen}
\boldsymbol{L}=\sum_{a=1}^v \lambda_{a}(\boldsymbol{q,p})\sum_{j_a=1}^{\sigma_a}\bigg(\frac{\partial}{\partial q^{a,j_a }}\otimes \rd q^{a,j_a }+\frac{\partial}{\partial p_{a,j_a }}\otimes \rd p_{a,j_a }\bigg)
 \end{equation}
 is a generator of $\mathscr{H}$, provided that $\{\lambda_{1}(\boldsymbol{q, p}), \ldots , \lambda_{v}(\boldsymbol{q, p})\}$ are arbitrary,
pointwise distinct smooth functions. Therefore, every operator $\boldsymbol{K}\in \mathscr{H}$ can be written in the form
 \begin{equation}\label{eq:KLagr}
 \boldsymbol{K}=\sum _{i=1}^m l_{i } \frac{\Pi_{j\neq i}(\boldsymbol{L}-\lambda_j \boldsymbol{I})}{\Pi_{j\neq i}(\lambda_i-\lambda_j )} \ , 
\end{equation}
where $l_i=l_i(\boldsymbol{q, p})$ are the eigenvalue fields of $\boldsymbol{K}$.
In particular, if
\begin{equation}\label{eq:64}
\lambda_a(\boldsymbol{q,p})=\lambda_a(q^{a,1},p_{a,1},\ldots, q^{a,\sigma_a}, p_{a,\sigma_a}) \quad a=1,\ldots,v\ ,
\end{equation}
the generator $\boldsymbol{L}$ is a Nijenhuis operator. This  operator endows the manifold M with a standard ωN structure, since it satisfies both conditions \eqref{eq:OmNa} and \eqref{eq:OmNd}, as it can be proved by a direct calculation
\par
If $v<m\leq n$,  by means of  a further decomposition of $\mathcal{V}_a$  we can re-obtain the case $m\leq v$. 

Finally, a generator $\boldsymbol{L}$ is maximal if and only if $m=v=n$. In this case,  $\boldsymbol{I}\in \mathscr{H}$.
\end{proposition}
\begin{proof}
 The inequality $v>n$ cannot hold in an $\omega \mathscr{H}$ manifold of class $n$ because,  due to Proposition \ref{pr:FW},  the rank of each distribution $\mathcal{V}_a$ can not be less than $2$. Therefore, we have that $v\leq n$. Let us consider a Darboux-Haantjes chart of the form \eqref{eq:HchartK},  adapted to the decomposition \eqref{eq:TMinters}  (whose existence is guaranteed by Theorem \ref{th:gHchart}). In this chart, if $v\geq m$, the generator \eqref{eq:Ldiag} takes the form \eqref{eq:Lgen} and it can be a Nijenhuis operator, providing that its eigenvalues fields are chosen of the form \eqref{eq:64}.
 
When $v<m$, a generator can still be constructed since, as in Proposition \ref{pr:CGKcoor},  we can further decompose
each  distribution $\mathcal{V}_a$ 
 into a direct sum of  $2$-dimensional sub-distributions
\[
\mathcal{V}_a=\bigoplus_{i=1}^{\sigma_a} Span\left\{ \frac{\partial}{\partial q^{a,i}}, \frac{\partial}{\partial p_{a,i}} \right\} \qquad a=1,\ldots,v \ .
\]
\end{proof}

\section{Complete Integrability and Haantjes structures}
\par
 The aim of this Section is to prove the  main result of this paper, namely the equivalence between the existence of an $\omega \mathscr{H}$ structure associated with a Hamiltonian system and its complete integrability in the sense of Liouville and Arnold. In particular, we shall prove formula \eqref{eq:LA}, which relates the Haantjes geometry of a given integrable system with its dynamics. Also, we will show in a specific example how the Haantjes formulation overcomes, for the vector fields under scrutiny, an obstruction to the existence of a classical Lenard chain pointed out by R. Brouzet.
 \par
From now on, we will work with the distinguished basis
$\{\tilde{\boldsymbol{K}}_1,\ldots,\tilde{\boldsymbol{K}}_m \}$, and we will drop off the $tilde$ over $\tilde{\boldsymbol{K}}_\alpha$ for the sake of simplicity.

 \subsection{Haantjes theorem for integrable systems}
 We propose a characterization of the notion of  integrability in the sense  of Liouville--Arnold in terms of    $\omega \mathscr{H}$ manifolds.
\begin{theorem}[Liouville-Haantjes]\label{th:HAL}
Let $M$ be a $2n$-dimensional Abelian $\omega \mathscr{H}$ manifold of class $n$ and $\{H_1,H_2, \ldots, H_n\}$ be smooth potential functions of a Haantjes chain generated by a function $H$.  Then, the foliation generated by these functions is Lagrangian. Consequently,  each Hamiltonian system, with Hamiltonian functions $H$ and $H_\alpha$, $1\le \alpha \le n$, is integrable by quadratures.
\par
Conversely, let us consider  a completely integrable system with $n$ degrees of freedom, defined by a Hamiltonian $H$ and a set of $n$  integrals of motion  $\{H_1,\ldots,H_n \}$, in involution  and independent among each other. Let $\{(J_k,\phi_k)\}$,   $k =1,\ldots, n$, denote a set of action-angle variables, with associated frequencies $\nu_{k}(\boldsymbol{J}):=\frac{\partial{H}}{\partial J_{k}}$. If $H$ is non degenerate, that is
\begin{equation}\label{eq:And}
\det\left(\frac{\partial \nu_k}{\partial J_i} \right)=\det\left(\frac{\partial^2 H}{\partial J_i \partial J_k} \right)\neq 0 \ ,
\end{equation}
then $M$ admits, in any tubular neighbourhood of an Arnold  torus, a semisimple $\omega \mathscr{H}$ structure  whose Haantjes algebra is generated by the operators
\beq \label{eq:LAA}
\boldsymbol{K}_\alpha=\sum _{i=1}^n \frac{\nu_i^{(\alpha)}(\boldsymbol{J})}{\nu_i (\boldsymbol{J})}\bigg (\frac{\partial}{\partial J_i}\otimes \rd J_i +\frac{\partial}{\partial \phi_i}\otimes \rd \phi_i \bigg )\quad \alpha=1,\ldots,n
\ ,
 \eeq
 where $\nu_i^{(\alpha)}(\boldsymbol{J})$ are the frequencies of the $(\alpha)-nth$ linear flow.
 \end{theorem}

\begin{proof}
 To prove the first statement, by virtue of  the classical Liouville-Arnold theorem, it is sufficient to note that the functions $H_\alpha$ belonging to a Haantjes chain are in involution w.r.t. the Poisson bracket defined by the symplectic form $\omega$, thanks to Lemma \ref{lm:MHchainInv}.

\par
Let us prove the converse statement. The integrals of motion $\{H_1,\ldots,H_n\}$
 are all assumed to be independent smooth functions on an open dense subset of the phase space, in involution among each others and with $H$. Due to the celebrated Arnold theorem \cite{AKN}, the $2n$-dimensional phase space is foliated by leaves whose compact connected components are invariant tori. Also, at least in any tubular neighbourhood of each torus, there exists a set of action-angle (AA) variables $\{( J_i,\phi_i)\}$, such that the symplectic 2-form reads
\beq \label{eq:omAA}
\omega=\sum_{i=1} ^n\rd J_i \wedge \rd \phi_i \ .
\eeq
Owing to  condition \eqref{eq:And},  the set $\{H_1,\ldots,H_n\}$ depends  on the action variables only  \cite{AKN}. Then, the functions $H_{\alpha}$  take  the generic form
\beq \label{eq:HJ}
H_\alpha=H_\alpha(\boldsymbol{J}), \qquad \alpha=1,\ldots,n.
\eeq
With these data, we shall construct a semisimple and Abelian  $\omega \mathscr{H}$  structure associated with the given integrable system.

As a basis of the Haantjes  algebra we wish to construct, we can take  the following \emph{diagonal} operators in the  action-angle coordinates
\beq \label{eq:Lkdiag}
\boldsymbol{K}_\alpha=\sum _{i=1}^n l_i^{(\alpha) }(\boldsymbol{J})\left(\frac{\partial}{\partial J_i}\otimes \rd J_i+
 \frac{\partial}{\partial \phi_i}\otimes \rd \phi_i \right) \qquad \alpha=1,\ldots,n\ ,
\eeq
where  $l_i^{(\alpha)}$ are arbitrary smooth functions. They comply with Definition \ref{def:HM} and fulfil the compatibility condition  \eqref{eq:compOmH}.

Moreover, we impose that the integrals of motion $\{H_1,H_2, \ldots,H_n\}$ form a Haantjes chain  generated by $H $, i.e.
\beq  \label{chaincond}
\boldsymbol{K}^{T}_{\alpha}dH=dH_{\alpha} \ , \qquad \alpha=1,\ldots,n\ .
\eeq
Being $\boldsymbol{K}_\alpha$ diagonal in the AA variables,  such conditions are equivalent to the following system of $2n$ \emph{algebraic} equations in the $n$ indeterminate functions: $l_i^{(\alpha)}$
\beq \label{eq:lJ}
l^{(\alpha)}_i \frac{\partial H}{\partial J_i}=  \frac{\partial H_{\alpha}}{\partial J_i}\ , \qquad \qquad i=1,\ldots,n.
\eeq
\beq\label{eq:lPhi}
 l^{(\alpha)}_i \frac{\partial H}{\partial \phi_i}=  \frac{\partial H_{\alpha}}{\partial \phi_i}\  , \qquad\qquad i=1,\ldots,n.
\eeq
Obviously, Eqs. \eqref{eq:lPhi} are trivially satisfied, so that only Eqs. \eqref{eq:lJ} have to be taken into account. Without loss of generality, we assume  that $\nu_i$, $i=1,\ldots,n$, are non-vanishing functions. Then,  equations \eqref{eq:lJ} imply that
  the eigenvalues of the $\alpha$-th operator must be the ratio between the frequencies associated to the $\alpha$-th integral of motion and to the Hamiltonian, respectively. It is easy to prove that the Haantjes operators so obtained are linearly independent, due to the independence of the integrals of motion.
Consequently,  the  Haantjes  algebra involved in the Haantjes chain \eqref{chaincond}  is generated by the distinguished basis of operators   \eqref{eq:LAA}.
\end{proof}

There is a natural relation between AA variables and DH coordinates in the Haantjes geometry, as clarified by
\begin{proposition}
Any set of AA variables for a completely integrable system is a set of DH coordinates for the $\omega\mathscr{H}$ manifold given by the symplectic form $ \omega$ and the Haantjes diagonal algebra generated by the operators \eqref{eq:LAA}.
\end{proposition}

\begin{remark} \label{rem:SIBH}
The Haantjes operators \eqref{eq:LAA} exist without any restriction on the form of the Hamiltonian function $H$, except for the non-degeneracy condition \eqref{eq:And}. However, if one wishes to construct a Nijenhuis recursion operator
$\boldsymbol{N}$ for $H$, i.e. a Nijenhuis operator that, at the same time, provides a classical Lenard chain
\beq
\rd H_\alpha=(\boldsymbol{N}^T)^{\alpha}\rd H\ \qquad\qquad \alpha=1,\ldots,n,
\eeq
and has  the natural vector fields $\left(\frac{\partial}{\partial J_i}, \frac{\partial}{\partial \phi_i}\right)$ as  eigenvectors, then the Hamiltonian function must take necessarily the separated  form
\beq \label{eq:HBH}
H(\boldsymbol{J})=\sum_{k=1}^n H_k(J_k) \ ,
\eeq
where $H_k(J_k) $ is a smooth function of the single action variable $J_i$ (see \cite{MM}, \cite{MK}).
\end{remark}

\begin{remark} \label{rem:LieLAA}
The non constant eigenvalues $l_i^{(\alpha)}(\boldsymbol{x})$ of the Haantjes operators \eqref{eq:LAA} $\boldsymbol{K}_\alpha$, $\alpha=1,2,\ldots,n$,  depending only on  action variables, are integrals of motion for the Hamiltonian vector field  $X_H$, i.e. their Lie derivatives along the flow of $X_{H}$ vanishes:
\[
\mathcal{L}_{X_{H}} l_{i}^{(\alpha)}=0.
\]

 However, this property does not imply that the Haantjes operators are invariant along the flow  of $X_H$, as
\beq \label{invKa}
\mathcal{L}_{X_H}\boldsymbol{K}_\alpha=
\sum_{i,k=1}^n \big (l_i^{(\alpha)}-l_k^{(\alpha)}\big )\frac{\partial \nu_i}{\partial J_k} \
\frac{\partial}{\partial \phi_i}\otimes \rd J_k \ .
\eeq
\end{remark}
\subsection {The analysis of Brouzet}
 In \cite{BrI},   R. Brouzet  studied the existence  of a Nijenhuis \textit{recursion} operator for   a completely integrable system, that is  a Nijenhuis operator \emph{compatible} with $\omega$ and fulfilling the requirement
\[
\mathcal{L}_{X_{H}}(\boldsymbol{N})=0\ .
\]
Notice that this requirement is not satisfied by the Haantjes operators \eqref{eq:LAA}, according to Eq. \eqref{invKa}.
Brouzet proved that the  existence of a Nijenhuis recursion  operator for $X_H$ in a tubular neighbourhood  of a Liouville torus implies very strong restrictions on the form of its Hamiltonian function. Accordingly, he presented an example of an integrable system with two degrees of freedom that \textit{does not admit} a recursion operator compatible with the original symplectic structure.  Here we show that such example does admit a simple formulation  in the context of the $\omega \mathscr{H}$ geometry.

In his analysis, Brouzet considered the symplectic manifold $M= \mathbb{R}^2\times \mathbb{T}^2$, with the action variables
${(J_1,J_2)}\in  \mathbb{R}^2 $, the angles ${(\phi_1,\phi_2)} $  on the bi-dimensional torus $ \mathbb{T}^2 $, and the Hamiltonian function
\beq \label{eq:HBr}
H=J_1(1+J_2^2)\ ,
\eeq
which is not of the form \eqref{eq:HBH} and  is non degenerate in the  dense open submanifold  $M':=\{m\in M: J_2\neq 0\}$.
The corresponding Hamiltonian vector field
\beq
X_H=(1+J_2^2)\frac{\partial}{\partial \phi_1}+2 J_1 J_2 \frac{\partial}{\partial \phi_2}
\eeq
is completely integrable, since  any smooth function depending only on the action variables is an integral of motion for it. For instance, let us take
\beq
H_1=J_1\ , \qquad \qquad H_2=J_2 \ ,
\eeq
which on $M'$ are functionally independent among each others.
One can easily verify that the two Hamiltonian functions in involution $\{H_1,H_2\}$ are the potential functions of  a Haantjes chain w.r.t. the $\omega \mathscr{H}$ structure given by the standard symplectic form

\beq
\omega=\rd J_1 \wedge \rd \phi_1+\rd J_2 \wedge \rd \phi_2
\eeq
and by the Haantjes operators
\beq
\boldsymbol{K}_1=\frac{1}{(1+J_2^2)} \left(\frac{\partial}{\partial J_1}\otimes \rd J_1+
\frac{\partial}{\partial \phi_1}\otimes \rd \phi_1 \right )\ , \quad
\boldsymbol{K}_2=\frac{1}{2 J_1J_2 }\left(\frac{\partial}{\partial J_2}\otimes \rd J_2 +
\frac{\partial}{\partial \phi_2}\otimes \rd \phi_2 \right )\ ,
\eeq
which are constructed in the  open submanifold of $M'$ where $J_1\neq 0$ according to  the formulae \eqref{eq:LAA}.

 It is interesting to observe that the authors of \cite{MLV} have by-passed the Brouzet obstruction to the existence of a Njenhuis recursion operator for the Hamiltonian \eqref{eq:HBr} (and for other examples presented in \cite{BrII}) by using a different strategy. Their approach consists in looking for a Njenhuis recursion operator compatible with a symplectic structure \textit{alternative} to the original one. Instead, in our theory, the Haantjes operators are compatible with the original symplectic structure.

\section{Integrable models, wave equation and Haantjes geometry}

Given the  equivalence between complete integrability of a Hamiltonian system and the existence of an associated  $\omega \mathscr{H}$ structure, we can use this equivalence in two ways: to construct integrable models from a given Haantjes geometry (the \emph{direct problem}) or, conversely, to
 determine the Haantjes geometry of a given integrable system (the \emph{inverse problem}).
In this section, we will adopt the first point of view, in order to show the flexibility of the Haantjes approach in applicative contexts. We consider the simplest case of a manifold $M$ of dimension $4$ and of a Haantjes algebra of class $2$, whose basis is  $\{\boldsymbol{I},\boldsymbol{K}\}$  for a suitable $\boldsymbol{K}$.   Indeed, by searching for Haantjes chains w.r.t.  such a  distinguished basis, we are able to define families of associated integrable models.

Precisely, we will show that by means of the Haantjes geometry, solutions of the two-dimensional wave equation can be used to define new integrable systems.
\begin{theorem}
Let $\xi=\frac{x+y}{\sqrt{2}}$, $\eta=\frac{x-y}{\sqrt{2}}$, $p_\xi=\frac{p_x+p_y}{\sqrt{2}}$, $p_\eta=\frac{p_x-p_y}{\sqrt{2}}$ be characteristic coordinates and momenta in an open set of $M$. The Hamiltonian
\beq
H=H_1(\xi,\eta,p_{\xi},p_{\eta})=f(\eta)+g(\xi)+F(p_{\eta})+ G(p_{\xi})
\eeq
where $f$, $g$, $F$, $G$ are arbitrary functions of their arguments, is integrable and admits the first integral of motion
\beq
H_2(\xi,\eta,p_{\xi},p_{\eta})=-f(\eta)+g(\xi)-F(p_{\eta})+ G(p_{\xi}).
\eeq
\end{theorem}

\begin{proof}
Consider the uniform Haantjes operator in Cartesian coordinates and momenta
\begin{equation} \label{eq:Wave}
\boldsymbol{K}= \frac{\partial}{\partial x}\otimes \rd y+\frac{\partial}{\partial y}\otimes \rd x+ \frac{\partial}{\partial p_x}\otimes \rd p_y+\frac{\partial}{\partial p_y}\otimes \rd p_x
 \ .
\end{equation}
We construct the Haantjes chain
\beq
\boldsymbol{K}^T dH= dH_2.
\eeq
This chain is defined by the differential relations
\bea
\begin{cases} \frac{\partial H_1}{\partial {p_y}}= \frac{\partial H_2}{\partial {p_x}},\\
 \frac{\partial H_1}{\partial {p_x}}=\frac{\partial H_2}{\partial {p_y}},
\end{cases} \qquad\qquad
\begin{cases} \frac{\partial H_1}{\partial y}= \frac{\partial H_2}{\partial x},\\
 \frac{\partial H_1}{\partial x}=\frac{\partial H_2}{\partial y} \ .
\end{cases}
\eea
These equations can be combined to define the wave equations
\beq
H_{i,p_x p_x}-H_{i,p_y p_y}=0, \qquad \qquad H_{i,xx}-H_{i,yy}=0, \quad i=1,2 \ .
\eeq
Therefore the Hamiltonian functions
\beq
H_1(x,y,p_x,p_y)=F(p_x-p_y)+G(p_x+p_y)+f(x-y)+g(x+y)
\eeq
and
\beq
H_2(x,y,p_x,p_y)=-F(p_x-p_y)+G(p_x+p_y)-f(x-y)+g(x+y) \ ,
\eeq
where $F,G,f,g$ are arbitrary smooth functions of their arguments, define a completely integrable system, separable in the coordinates $(\xi, \eta, p_{\xi}, p_{\eta})$.
\end{proof}
\begin{example}
Choosing the functions $F,G,f,g$ as a power of their arguments, we get the interesting family of models
\begin{eqnarray}
H_1&=&(p_x-p_y)^n +(p_x+p_y)^n+(x-y)^m+(x+y)^m  \ ,\quad n,m\in \mathbb{N}\\
H_2&=&-(p_x-p_y)^n +(p_x+p_y)^n-(x-y)^m+(x+y)^m \ .
\end{eqnarray}
For $n=2$, the  Hamiltonian function $H_1$ is quadratic in the momenta and corresponds to a class of separable systems that have been discussed in \cite{Pere} (page 81).
In particular, for $n=2,m=3$ one obtains the Sawada-Kotera  system \cite{AzSa}. For $n>2$ we have, to the best of our knowledge, a new family of  integrable systems.
\end{example}
The direct method outlined in this section can be widely adopted to generate new models from known Haantjes operators. However, an exhaustive analysis of this approach is out of the scopes of the present work.

\section{The inverse problem for  systems with two degrees of freedom} \label{Sec:Proc}
In this Section, we deal  with the inverse problem. More precisely, given a set of independent functions  in involution, we  will construct by means of them a Haantjes algebra \textit{compatible with the symplectic form}. In other words, we shall determine a Haantjes algebra for an  assigned integrable system, represented in physical variables.
\subsection{A  general procedure} \label{SubSec:Proc}
Let us consider the simplest case of Hamiltonian systems with two degrees of freedom. We propose a general procedure allowing us to determine   a   Haantjes algebra with identity and of class $2$, associated to the Haantjes chain formed by the differentials of the two given  integrals of motion.
 We search for a generator of such Haantjes algebra, that is to say a Haantjes operator $\boldsymbol{L}$ whose minimal polynomial is of degree two (namely, the maximum  degree allowed by our assumptions):
\begin{equation}\label{eq:minpol2}
m_{\boldsymbol{L}}(\boldsymbol{x},\lambda):= \lambda^2-c_1(\boldsymbol{x}) \lambda-c_2(\boldsymbol{x})  \ ,
\end{equation}
where $c_1(\boldsymbol{x}) =\frac{1}{2} Trace(\boldsymbol{L})$, $c_2(\boldsymbol{x}
)=-\sqrt{ det( \boldsymbol{L} )}$.
Let us note that such a requirement does not imply the semisimplicity of $\boldsymbol{L}$ (unless the existence of two real, \emph{distinct} roots of $m_{\boldsymbol{L}}(\boldsymbol{x}, \lambda)$ is also assumed).
\par
\begin{remark} \label{rem:Hmod2gl}
 In the case $n=2$,  any non-isotropic (that is   $\boldsymbol{L}_1\neq l(\boldsymbol{x})\boldsymbol{I}$) Haantjes  operator,  compatible with
the symplectic form, is a generator of the cyclic Haantjes algebra $\mathcal{L}(\boldsymbol{L}_1)=Span\{\boldsymbol{I},\boldsymbol{L}_1\}$. Besides, any other generator of $\mathcal{L}(\boldsymbol{L}_1)$ has the form

 \begin{equation}
\boldsymbol{L}_2=f \boldsymbol{I}+g\boldsymbol{L}_1 \ ,
 \end{equation}
 where $f$ and $g$ are arbitrary smooth functions, with $g$ nowhere vanishing. In fact,
\begin{equation}
det(\boldsymbol{L}_2-\lambda\boldsymbol{I})=det(f \boldsymbol{I}+g\boldsymbol{L}_1-\lambda\boldsymbol{I})=g^n\,det\left(\boldsymbol{L}_1-\frac{\lambda-f}{g}\boldsymbol{I}\right) \ .
\end{equation}
Therefore, the eigenvalues $\lambda_i^{(1)}$ of $\boldsymbol{L}_1$ and $\lambda_i^{(2)}$ of $\boldsymbol{L}_2$ are related by the affine equations
\begin{equation} \label{eq:Nautov2gl}
\lambda_i^{(2)}=f+g\, \lambda_i^{(1)} \qquad i=1,2 \ ;
\end{equation}
consequently, $\lambda_1^{(1)}=\lambda_2^{(1)} \Leftrightarrow \lambda_1^{(2)}= \lambda_2^{(2)}$. Then, we can conclude that the Haantjes algebra (even if non-semisimple) contains a maximal generator if and only if all of its generators are maximal.
\end{remark}

 The procedure can be sketched as follows.
Given two independent integrals of motion in involution $\{H_1=H, H_2\}$, we look for an operator $\boldsymbol{L}$  with the following properties:
\\
\begin{itemize}
\item[i)]
it is compatible with the symplectic form
\begin{equation}
 \label{eq:LOm2gl}
\boldsymbol{L}^T\boldsymbol{\Omega}=\boldsymbol{\Omega}\,\boldsymbol{L} \ ;
\end{equation}
\item[ii)]
it provides us with a Haantjes chain for the integrals
\begin{equation}
\label{eq:Hc2gl}
\boldsymbol{L}^T \rd H =\rd H_2 \ ;
\end{equation}
\item[iii)]
it is a Haantjes operator
\begin{equation}
\label{eq:HaanNull2gl}
\mathcal{H}_{\boldsymbol{L}}(X, Y)= \boldsymbol{0} \qquad \forall X, Y \in TM \ .
\end{equation}
\end{itemize}
The algebraic compatibility condition \eqref{eq:LOm2gl} is very strong: chosen a set of Darboux coordinates, it  allows us to reduce the number of unknown components  of the operator $\boldsymbol{L}$  from $16$ to $6$. We obtain that it must have the form
\begin{equation} \label{eq:L2g}
\boldsymbol{L}=
\left[
\begin{array}{cc|cc}
 l_{1} ^1& l_{2}^1&0&l_{4}^1 \\
     l_{1}^2& l_{2}^2&-l_{4}^1 &0 \\
     \hline
      0&l_{2}^3&l_{1}^1&l_{1}^2 \\
    -l_{2} ^3& 0 & l_{2} ^1& l_{2} ^2\\
      \end{array}
      \right ] \ ,
\end{equation}
where $l_{j}^i$ are unknown arbitrary  functions on $M$. Note that the form \eqref{eq:L2g} for $\boldsymbol{L}$ guarantees that condition \eqref{eq:minpol2} is satisfied.
The relations \eqref{eq:Hc2gl}, still algebraic, provide us with a system of $4$ algebraic equations in the  $6$ unknown functions; it turns out that only $3$ equations are independent; thus we are left with $3$ unknown functions.
 The  vanishing of the Haantjes torsion \eqref{eq:HaanNull2gl}  of $\boldsymbol{L}$  provides us with an over-determined system of $24$ PDEs of first order,   which can be managed with some suitable ans\"atze.  For instance,  some homogeneity properties for the components of $\boldsymbol{L}$ can be assumed.
 \par
\vspace{2mm}

\subsection{On the superintegrable Post-Winternitz system}\label{Sec:PW}
By means of the procedure described above,
we can discuss now the inverse problem for  a system which recently has attracted much attention: the Post-Winternitz (PW) system \cite{PW}. Indeed,  it is a  maximally \emph{superintegrable} system \cite{MPW} with  integrals of motion  cubic and quartic in the momenta. As a consequence of maximal superintegrability, its bounded orbits are closed and periodic. Thus, as any superintegrable system, it does not fulfil the non degeneracy condition \eqref{eq:And}, so that Theorem \ref{th:HAL} cannot be applied.  Despite  its regularity properties, separation variables for the PW system  are not known. Since it does not belong to the classical St\"ackel  family of Hamiltonian functions quadratic in the momenta, the PW system is certainly not separable by  an extended point transformation.

Let us consider a set of canonical coordinates $(x,y,p_x,p_y)$; the Hamiltonian system with Hamiltonian function
\begin{equation} \label{eq:HPWfis}
H=H_1=\frac{1}{2}(p_x^2+p_y^2)+a \frac{x}{y^\frac{2}{3}} \qquad a\in \mathbb{R}\ , \qquad y\neq 0 \ ,
\end{equation}
admits  the two independent integrals of motion
\begin{equation} \label{eq:H2PWfis}
H_2=2p_x^3+3p_y^2p_x +a\left(9 {y^\frac{1}{3}} p_y +6 \frac{x}{y^\frac{2}{3}}p_x \right)
\end{equation}
and
\begin{equation} \label{eq:H3PWfis}
H_3=p_y^4-12 a  y^\frac{1}{3}p_x p_y +4 a \frac{x}{y^\frac{2}{3}}p_y^2-2 a^2 \left(9y^{2/3}-\frac{2x^2}{y^\frac{4}{3}}\right) \ .
\end{equation}
We shall prove that these integrals form the two different Haantjes chains $\{\rd H_1,\rd H_2\}$ and $\{\rd H_1,\rd H_3\}$; each of them is sufficient to guarantee the complete integrability of the PW system.

By performing the extended-point canonical transformation
\begin{equation}\label{eq:tcPW}
q_1=y^\frac{1}{3} \ , \quad q_2=\frac{x}{y^\frac{2}{3}} \ , \quad
 p_1=2\frac{x}{y^\frac{1}{3}}p_x+3y^\frac{2}{3}p_y\ , \quad p_2=y^\frac{2}{3}p_x \ ,
\end{equation}
we reduce the Hamiltonian functions to a rational form; from it we infer the weights of the  three  components of \eqref{eq:L2g} (still unknown after having imposed the conditions \eqref{eq:Hc2gl}).
As a result of the previous approach, we get the $\omega\mathscr{H}$  manifold $(\omega,\boldsymbol{I},\boldsymbol{K}_{2}^{(PW)}   )$  associated with the Haantjes chain $\{\rd H,\rd H_2\}$, where
\begin{equation} \label{eq:L2}
\boldsymbol{K}_{2}^{(PW)}=3
\left[\begin{array}{cc|cc}
2p_x& p_y & 0 & 3y \\
0 &2p_x &-3y& 0 \\ \hline
0 & 0 & 2p_x&0\\
0 & 0 &  p_y  & 2p_x
\end{array}
\right] \ ,
\end{equation}
once written in the original coordinates $(x,y,p_x,p_y)$. Similarly,  we obtain a second $\omega\mathscr{H}$  structure $(\omega,\boldsymbol{I},\boldsymbol{K}_{3}^{(PW)}   )$ associated with the Haantjes chain $\{\rd H,\rd H_3\}$, where
\begin{equation} \label{eq:L3}
\boldsymbol{K}_{3}^{(PW)}=4
\left[
\begin{array}{cc|cc}
p_y^2+2a\frac{x}{y^{2/3}}& -(p_x p_y+3ay^{1/3} )& 0 & -3y p_x \\
0&p_y^2+2a\frac{x}{y^{2/3}} &3yp_x& 0 \\ \hline
0 & 0 & p_y^2+2a\frac{x}{y^{2/3}} &0\\
0 & 0 &-(p_x p_y+3ay^{1/3} )  &p_y^2+2a\frac{x}{y^{2/3}}
\end{array}
\right] \ .
\end{equation}

Although both $\boldsymbol{K}_{2}^{(PW)}$ and $\boldsymbol{K}_{3}^{(PW)}$ have a minimal polynomial of degree $2$, they are not semisimple, since each of them possesses a sigle eigenvalue of algebraic multiplicity equal to $4$, with two proper eigenvectors and two generalized eigenvectors.

Moreover, neither of the two non-semisimple Haantjes algebras generated by $\{ \boldsymbol{I},\boldsymbol{K}_{2}^{(PW)}\}$ and $\{ \boldsymbol{I},\boldsymbol{K}_{3}^{(PW)}\}$ admits a maximal generator, by Remark \ref{rem:Hmod2gl}. However, their two Haantjes chains assure the superintegrability of the PW model. Thus, the existence of  $\boldsymbol{K}_{2}^{(PW)}$ and $\boldsymbol{K}_{3}^{(PW)}$ shows that the Haantjes theory can be naturally applied  to non-St\"ackel systems which do not possess any  evident separability structure, even when they do not satisfy the non-degeneracy condition \eqref{eq:And}. Furthermore, one can verify that the set  $\{ \boldsymbol{I},\boldsymbol{K}_{2}^{(PW)}, \boldsymbol{K}_{3}^{(PW)}\}$ is still a basis of a Haantjes algebra. However, this algebra is non-Abelian, therefore it does not admit a generator. Thus, the PW system is a remarkable example of a superintegrable system admitting a non Abelian  $\omega\mathscr{H}$ structure of rank $3$.

\begin{remark}
\noi Once the algebra $\mathscr{H}$ associated with an integrable system is determined, assuming that it is Abelian a theorem proved in \cite{TT2017} guarantees that there exists a local chart where all Haantjes operators can be put in diagonal form (if $\mathscr{H}$ is semisimple) or in a block-diagonal form (if it is not). This fact is crucial in order to find separation variables \cite{RTT2020}, whenever they exist, or, more generally, to study partial separability \cite{Stackel1897}, \cite{CR2019}. 
\end{remark}
\begin{remark}
The procedure described above can also be applied to the case when an Hamiltonian function is given, but no integrals of motion are known. In this situation, Eq. \eqref{eq:Hc2gl} for the Haantjes chain can be used to construct, in principle, both the integrals of motion (if they exist) and the corresponding Haantjes operators.
\end{remark}

\section{The stationary reduction of the seventh-order KdV flow revisited}
In this section, in order  to show the large range of applicability of the theory previously developed, we shall discuss an important example of Hamiltonian integrable system, defined on a six-dimensional symplectic manifold, which is obtained as a stationary reduction of the  seventh-order equation  of the Korteweg de Vries (KdV) hierarchy.

In \cite{TKdV}, a method to obtain the  Poisson pencil   $P_1-\lambda P_0$ of the stationary flows of the KdV hierarchy was presented. In \cite{MT}, this method was applied to construct the stationary reduction of the  seventh-order equation of the hierarchy.   The restricted Poisson pencil turns out to be a degenerate pencil of co-rank one in a seven dimensional manifold $\mathcal{M}^{(7)}$, being therefore  a Gelfand-Zakarevich system \cite{GZ}. It possesses a polynomial Casimir function of length four, starting with a Casimir of $P_0$ and ending with a Casimir of $P_1$. Then,  a Marsden-Ratiu reduction procedure \cite{MR},   similar to
the one used in other  cases \cite{TKdV,MTlt,MTltC}, was performed   to each
six-dimensional symplectic leaf   $S_0$ of the Poisson tensor $P_0$, in order to get   rid of the Casimir  of $P_0$.

\noi Furthermore, by restricting the  polynomial Casimir function to $S_0$, one of the present authors  was able to obtain in \cite{MT}
three Hamiltonian functions in involution which in the Darboux chart $(q_1,q_2,q_3,p_1,p_2,p_3)$ read
\begin{align}\label{eq:HKdV7}
H=H_1&=p_1p_2+\frac{1}{2}p_3^2-\frac{5}{8}q_1^4+\frac{5}{2}q_1^2q_2+\frac{1}{2}q_1q_3^2-\frac{1}{2}q_2^2, \\
\nonumber
H_2&=\frac{1}{2}p_1^2
+p_1p_2q_1+p_3^2q_1-p_2^2q_2-p_2p_3q_3-\frac{1}{2}q_1^5-\frac{1}{4}q_1^2q_3^2+\frac{1}{2}q_2q_3^2+2q_1q_2^2, \\
\nonumber
H_3&=\frac{1}{2}p_3^2q_1^2+p_3^2q_2-p_1p_3q_3-p_2p_3q_1q_3+\frac{1}{2}p_2^2q_3^2+\frac{1}{2}q_1^3q_3^2-q_1q_2q_3^2-\frac{1}{8}q_3^4.
\end{align}

 However, as typically happens in the case of Gelfand-Zakarevich systems,  the reduced integrable Hamiltonian  systems on $S_0$ do not possess a bi-Hamiltonian description but a $\omega N$ one \cite{FP}.
 Nevertheless, they can also be described in the context of our new theory.
In  fact, we search for a generator $\boldsymbol{L}$ of a cyclic Haantjes   algebra $\mathscr{H}$ of rank $3$, therefore with the minimal polynomial  of degree $3$:
\begin{equation}
m_{\boldsymbol{L}}(\boldsymbol{x},\lambda):= \lambda^3-c_1(\boldsymbol x) \lambda^2-c_2(\boldsymbol x)\lambda-c_3(\boldsymbol x) \ .
\end{equation}
To this aim, we follow a procedure whose first step is analogous to the one performed in two degrees of freedom. We  look for an operator $\boldsymbol{K}_2$ that satisfies
\begin{align}
 \label{eq:LOm3gl}
\boldsymbol{K}_2^T\boldsymbol{\Omega}&=\boldsymbol{\Omega}\,\boldsymbol{K}_2 \ ,\\
\label{eq:Hc3gl}
\boldsymbol{K}_2^T\rd H &=\rd H_2 \ ,\\
\label{eq:HaanNull3gl}
\mathcal{H}_{\boldsymbol{K}_2}(X, Y)&= \boldsymbol{0} \qquad \forall X, Y \in TM \ .
\end{align}
Condition \eqref{eq:LOm3gl} allows us to reduce the unknown components  of the  operator $\boldsymbol{K}_2$ from $36$ to $15$.
Under the simplest ansatz  that the remaining elements of  $\boldsymbol{K}_2$ are \emph{linear} in the Darboux coordinates $(q_1,q_2,q_3,p_1,p_2,p_3)$,
we find the following, unique solution of Eqs. \eqref{eq:LOm3gl}, \eqref{eq:Hc3gl} and \eqref{eq:HaanNull3gl}:
\begin{equation} \label{eq:L1K}
\boldsymbol{K}_2=
\left[
\begin{array}{ccc|ccc}
    q_1 &1&0&0 &0&0\\
 -2q_2&q_1&-q_3&0& 0&0\\
 -q_3&0&2q_1&0&0&0\\
  \hline
    0& -p_2\ & -p_3\ &q_1\ &-2q_2\ &-q_3\\
     p_2&0 & 0& 1&q_1&0\\
   p_3&0&0 &0&-q_3&2q_1
      \end{array}
      \right ]\ .
\end{equation}
Since  $\boldsymbol{K}_2$  is a maximal semisimple Haantjes operator in the points where the discriminant $\Delta$ of its minimal polynomial 
\begin{equation}
\Delta=-(8\, q_1^4q_2+4\,q_1^3q_3^2+32\,q_1^2q_2^2+72\,q_1q_2q_3^2+27\,q_3^4+32\,q_2^3)
\end{equation}
is positive, then, by virtue of Proposition \ref{pr:Hmg},  it generates a cyclic Haantjes algebra $\mathscr{H}$. Thus, we search for another Haantjes operator $\boldsymbol{K}_{3}$ such that
\begin{align}
\label{eq:KdVK2}
\boldsymbol{K}_3&=f \boldsymbol{I}+g\boldsymbol{L} +h\boldsymbol{L}^2 \ , \quad
\boldsymbol{L}=\boldsymbol{K}_2 \ ,\\
\label{eq:KdVLHc}
\boldsymbol{K}_3^T\rd H&=\rd H_3 \ ,
\end{align}
where $f$, $g$, $h$ are suitable smooth functions on $M$ to be determined. The unique solution is
  $\boldsymbol{K}_3=(q_1^2+2q_2)\,\boldsymbol{I}-2q_1\boldsymbol{L}+\boldsymbol{L}^2$, therefore
\begin{equation} \label{eq:L2K}
\boldsymbol{K}_{3}=
\left[
\begin{array}{ccc|ccc}
0 &0&-q_3&0&0&0\\
   q_3^2&0&-q_1q_3&0&0&0\\
    -q_1q_3\ &-q_3\ &q_1^2+2q_2&0 &0&0\\
    \hline
 0& 0&p_2q_3-p_3q_1&0&q_3^2&-q_1q_3 \\
  0& 0& -p_3& 0 &0&-q_3\\
  -(p_2q_3-p_3q_1) \ &p_3&0 &  -q_3\ & -q_1q_3\ & q_1^2+2q_2
      \end{array}
      \right] \ .
\end{equation}
Thus, $\{\boldsymbol{K}_1=\boldsymbol{I}, \boldsymbol{K}_2, \boldsymbol{K}_3\}$ is a distinguished basis of the cyclic Abelian Haantjes algebra of rank $3$ generated by $ \boldsymbol{L}=\boldsymbol{K}_2$, which provides us with the Haantjes chain $\{ \rd H_1,\rd H_2, \rd H_3\}$.

\section{Future Perspectives}

The extension of the present theory to the case of quantum integrable systems is a nontrivial task. This research line would pave the way to an algebraic interpretation
 of the notion of Haantjes integrability developed here, in terms of infinite-dimensional commuting operators on a Hilbert separable space of quantum states.

Also, it would be interesting to compare the geometric structures underlying the vision offered here with the intrinsic, purely algebraic structures developed in \cite{IMRT}, in the context of \textit{nilpotent integrability}.

We mention that new $\omega\mathscr{H}$ structures have been recently found in Ref. \cite{Giapp},  which is based on an earlier version of this article.  

An in-depth analysis of the case of \textit{superintegrable systems} \cite{TWR}, especially maximally superintegrable ones, has been performed \cite{RTT2020}. Along these lines, we also wish to construct  a generalization of our approach to the study of the geometry of certain classical systems, as the abovementioned Post-Winternitz model of Section \ref{Sec:PW}, that do not possess any known system of separation coordinates. We believe that our theory can offer a proper language in which the study of separability can be formulated and carried out.

\section*{Acknowledgement}

\noi The research of P. T. has been supported by the research project
PGC2018-094898-B-I00, Ministerio de Ciencia, Innovaci\'{o}n y Universidades and Agencia Estatal de Investigaci\'on,
Spain, and by the Severo Ochoa Programme for Centres of Excellence in R\&D
(CEX2019-000904-S), Ministerio de Ciencia, Innovaci\'{o}n y Universidades y Agencia Estatal de Investigaci\'on, Spain. 

\noi  P. T. is a member of Gruppo Nazionale di Fisica Matematica (GNFM) of INDAM.

\end{document}